%% file: LMCS-final-arxiv-code.tex
\documentclass{LMCS}

\def\dOi{9(4:2)2013}
\lmcsheading%
{\dOi}
{1--23}
{}
{}
{Nov.~14, 2012}
{Oct.~\phantom08, 2013}
{}

\ACMCCS{[{\bf Theory of computation}]: Computational complexity and
  cryptography---Complexity theory and logic} 

\usepackage{hyperref}
\newcommand{\deffi}{:=}
\renewcommand{\defi}{\deffi}

    \newcommand{\Kommentar}[1]{}
    
    \newcommand{\imp}{\Longrightarrow} 
    \renewcommand{\epsilon}{\varepsilon}

\allowdisplaybreaks

    \newcommand{\pos}{\operatorname{\mathfrak{p}}}
    
    \newcommand{\cl}{\operatorname{cl}}
    \newcommand{\crit}{\operatorname{crit}}
    \newcommand{\bd}{\operatorname{bd}}
    \newcommand{\choice}{\kappa}
    \newcommand{\appl}{\operatorname{appl}}
    \newcommand{\startpos}{\operatorname{\mathfrak{s}}}
    \newcommand{\whitenode}{\tikz{\node[circle,draw=black,fill=white,inner  sep=0pt,minimum  size=1mm] at (0,0){};} }
    \newcommand{\blacknode}{\tikz{\node[circle,draw=black,fill=black,inner  sep=0pt,minimum  size=1mm] at (0,0){};} }
    \newcommand{\xnode}{\mathsf x}
    \newcommand{\ynode}{\mathsf y}
    \renewcommand{\circ}{\uplus}
    \renewcommand{\bigcirc}{\biguplus}
    
    \newcommand{\DTIME}{{\upshape DTIME}}
    \newcommand{\LOGSPACE}{{\upshape LOGSPACE}}
    \newcommand{\goal}{\gamma}

    \newcommand{\dom}{\operatorname{Dom}}

    \newtheorem{theorem}{Theorem}
    \newtheorem{lemma}[theorem]{Lemma}

    \theoremstyle{definition}
    \newtheorem{definition}[theorem]{Definition}
    \newtheorem{claim}[theorem]{Claim}

    \newcommand{\innerqed}{$\dashv$}
    \newcommand{\outerqed}{\qedsymbol}
    \newenvironment{innerproof}{\renewcommand{\qedsymbol}{\innerqed}\begin{proof}}{\end{proof}\renewcommand{\qedsymbol}{\outerqed}}
    
    \newcounter{todos}
    \usepackage{ifthen}

\newenvironment{IEEEproof}{\proof}{\qed}
    \usepackage{tikz}
      \usetikzlibrary{arrows}
      \usetikzlibrary{decorations.pathmorphing}
      \usetikzlibrary{decorations.shapes}
      \usetikzlibrary{backgrounds}
      \usetikzlibrary{positioning}
      \usetikzlibrary{fit}
      \usepackage{pgffor}
  
  \definecolor{darkgreen}{rgb}{0,0.4,0} 
\usepackage{url}

\theoremstyle{plain}
\def\eg{{\em e.g.}}

\begin{document}

\title{Lower Bounds for Existential Pebble Games and $k$-Consistency Tests}

\author{Christoph Berkholz}	
\address{Lehrstuhl f\"ur Informatik 7, RWTH Aachen University}	
\email{berkholz@informatik.rwth-aachen.de}  



\titlecomment{{\lsuper*}This is the full version of the conference paper \cite{Berkholz.2012}.}


\begin{abstract}
The existential $k$-pebble game characterizes the expressive power of the ex\-is\-ten\-tial-positive $k$-variable fragment of first-order logic on finite structures. 
The
winner of the existential $k$-pebble game on two given finite structures
can be determined in time $O(n^{2k})$ by dynamic programming on the graph of game configurations.  We show that there is no
$O(n^{(k-3)/12})$-time algorithm that decides which player can win the
existential $k$-pebble game on two given structures. This lower bound is
unconditional and does not rely on any complexity-theoretic
assumptions.

Establishing strong $k$-consistency is a well-known heuristic for solving the constraint satisfaction problem (CSP). By the game characterization of Kolaitis and Vardi \cite{Kolaitis.2000a} our result implies that there is no $O(n^{(k-3)/12})$-time algorithm that decides if strong $k$-consistency can be established for a given CSP-instance.
\end{abstract}

\maketitle

\section{Introduction}

For two finite relational structures \textbf{A} and \textbf{B} the homomorphism problem asks if there is a mapping from the domain $A$ of \textbf{A} to the domain $B$ of \textbf{B} that preserves all relations. As pointed out by Feder and Vardi \cite{Feder.1998} this problem is equivalent to the constraint satisfaction problem (CSP) where the variables correspond to the domain of \textbf{A}, the values correspond to the domain of \textbf{B} and the constraints are encoded in the relations of \textbf{A} and \textbf{B}. Thus, every  homomorphism from \textbf{A} to \textbf{B} corresponds to a solution of the CSP. Since the homomorphism problem and the CSP are NP-complete in general, there is need to look for heuristics. One well-known method introduced in the context of constraint satisfaction is the procedure of establishing strong $k$-consistency, which can be implemented by an $O(n^{2k})$-time algorithm (see \eg{} \cite{Cooper.1989,Kolaitis.2000a,Atserias.2007}). We will explain this concept below in the setting of the homomorphism problem.

From a logical point of view, there is a homomorphism between two finite structures \textbf{A} and \textbf{B} if and only if every existential-positive first-order sentence that is true on \textbf{A} is also true on \textbf{B}. Instead of considering the full existential-positive fragment of first-order logic one can relax that question and ask whether every $k$-variable existential-positive first-order sentence true on \textbf{A} is also true on \textbf{B}. Such questions can be analyzed using combinatorial games. For the $k$-variable existential-positive fragment of first-order logic the corresponding game is the existential $k$-pebble game \cite{Kolaitis.1995}, which is defined in Section \ref{sec:definitions}.

Kolaitis and Vardi \cite{Kolaitis.2000a} showed that this logical relaxation of the homomorphism problem is equivalent to the $k$-consistency heuristic. That is, strong $k$-consistency can be established if and only if Duplicator wins the existential $k$-pebble game on \textbf{A} and \textbf{B}. To sum up, the following three statements are equivalent on finite relational structures \textbf{A} and \textbf{B}:
\begin{itemize}
 \item Strong $k$-consistency can be established.
 \item For every existential-positive $k$-variable first-order sentence $\varphi$:  $\textbf{A}\models \varphi \imp \textbf{B}\models\varphi$.
 \item Duplicator has a winning strategy in the existential $k$-pebble game. 
\end{itemize}
Now we state our main result that provides a lower bound on the computational complexity of the statements above.
\begin{theorem} \label{maintheorem} For every fixed $k\geq 15$ and any $\epsilon>0$, the winner of the existential $k$-pebble game on two given finite relational structures \textbf{A} and \textbf{B} cannot be determined in time $O((\|\textbf{A}\|+\|\textbf{B}\|)^{\frac{k-2}{12}-\epsilon})$ on deterministic multi-tape Turing machines.
\end{theorem}
From this theorem we directly get an $\Omega(n^{\frac{k-2}{12}-\epsilon})$ lower bound for deciding if strong $k$-consistency can be established, where $n$ is the size of the CSP-instance. 
As an upper bound, the query ``Does Spoiler win the existential $k$-pebble game on \textbf{A} and \textbf{B}?'' is LFP$^{2k}$-definable  for two given structures \textbf{A} and \textbf{B} \cite{Kolaitis.2000} and is decidable by an $O(|A|^k|B|^k)$-time algorithm. 
We prove Theorem \ref{maintheorem} by a reduction from the $k$-pebble game of Kasai, Adachi and Iwata \cite{Kasai.1978}, called KAI-game, to the existential ($k+1$)-pebble game. Our result then follows from an $n^{\Omega(k)}$ lower bound for this game \cite{Adachi.1984}, which in turn follows from the deterministic time hierarchy theorem.

\subsection{Related Work}

Kolaitis and Panttaja \cite{Kolaitis.2003} proved that for every fixed $k\geq 2$ the problem of determining the winner of the existential $k$-pebble game is complete for PTIME under LOGSPACE reductions. Furthermore, they established that the problem is complete for EXPTIME when $k$ is part of the input. It follows that there is no algorithm for this problem whose running time is polynomial in the size of the structures as well as in the number of pebbles.
Parameterized by the number of pebbles $k$, the problem is known to be W[1]-hard. This follows directly from the fact that a graph $G$ contains a $k$-clique if and only if Duplicator has a winning strategy for the existential $k$-pebble game on the complete graph on $k$ vertices and $G$. 
Thus, the existence of an algorithm of running time $f(k)n^{c}$ for some computable function $f$ and constant $c$ would imply W[1] = FPT, an unlikely event in parameterized complexity theory.
However, since we do not know whether W[1] = FPT it is consistent with our previous knowledge that there exists an $O(2^{k}n^{2})$ algorithm determining the winner of the existential $k$-pebble game on two relational structures. Thus, for every fixed $k$, it was possible that there exists a quadratic time algorithm deciding if strong $k$-consistency can be established.

To prove the EXPTIME-completeness Kolaitis and Panttaja reduced the KAI-game to the existential pebble game. In this reduction the number of pebbles used in the existential pebble game depends on the size of the KAI-game instance and is not bounded by any function of the number of pebbles used in the KAI-game. Thus, their reduction fails to prove a lower bound for fixed $k$ and it was left as an open question if such a lower bound can be proven. 
In this paper we reduce the $k$-pebble KAI-game to the existential ($k+1$)-pebble game, and thus keep the parameter small. Some constructions are quite similar to those used by Kolaitis and Panttaja. However, the proof differs significantly at crucial points. Furthermore, to devise the winning strategy for Duplicator we use a different proof technique. 

Finally, the parameterized complexity of $k$-consistency has also been investigated by Gaspers and Szeider \cite{Gaspers.2011}. We discuss their work after the introduction to $k$-consistency in Section \ref{sec:establ}.

\subsection{Further Applications in Finite Model Theory}

We can improve Theorem \ref{maintheorem} in two ways. 
First, all partial homomorphisms in Duplicator's winning strategy (defined below) are in fact partial isomorphisms. 
Thus, Duplicator has a winning strategy in the $k$-pebble game that corresponds to the existential $k$-variable fragment of first-order logic, where negation is allowed in front of atomic formulas. 
This implies that it requires $\Omega((\|\textbf{A}\|+\|\textbf{B}\|)^{\frac{k-2}{12}-\epsilon})$ time to decide if every existential $k$-variable first-order sentence true on \textbf{A} is also true on \textbf{B}. 
Second, the structures constructed in our reduction are directed graphs. 
Therefore, Theorem \ref{maintheorem} holds even when we restrict ourselves to $\sigma$-structures, where $\sigma$ is a relational signature containing at least one binary relation. 

\subsection{Establishing Strong \texorpdfstring{$k$}{k}-Consistency}\label{sec:establ}

To fix the terminology, we briefly introduce the concept of establishing strong $k$-consistency as it is defined in \cite{Kolaitis.2000a}.
Let \textbf{A} and \textbf{B} be two finite relational structures with universes $A$ and $B$. A \textit{$k$-partial homomorphism} is a partial homomorphism with domain size $k$. 
\textbf{A} and \textbf{B} are \textit{$k$-consistent} if for every ($k-1$)-partial homomorphism $h$ from \textbf{A} to \textbf{B} and every $a\in A$ there is a partial homomorphism that extends $h$ and is defined on $a$. Two structures are \textit{strongly $k$-consistent} if they are $i$-consistent for every $i\leq k$.

Strong $k$-consistency \textit{can be established} for \textbf{A} and \textbf{B} if there are two strongly $k$-consistent structures \textbf{A'} and \textbf{B'} over the same universes $A$ and $B$ such that the following two statements hold.
\begin{itemize}
\item Every $k$-partial homomorphism from \textbf{A'} to \textbf{B'} is a $k$-partial homomorphism from \textbf{A} to \textbf{B}.
\item Every  function from A to B is a homomorphism from \textbf{A'} to \textbf{B'} if and only if it is a homomorphism from \textbf{A} to \textbf{B}.
\end{itemize}
Loosely speaking, strong $k$-consistency can be established for two structures if they can be made strongly $k$-consistent by adding new relations and without changing the solution space with respect to the homomorphism problem. It is easy to see that if there is a homomorphism from \textbf{A} to \textbf{B}, then strong $k$-consistency can be established. Although the converse is not true in general, it holds for some classes of structures \cite{Atserias.2007,Dalmau.2002}. 

All known $k$-consistency algorithms, as \eg{} \cite{Cooper.1989}, iteratively propagate new constraints (relations in our notion) until the instance becomes $k$-consistent. Gasper and Szeider \cite{Gaspers.2011} argued that the task of checking whether the instance is already $k$-consistent is inherent in this procedure and thus lower bounds its complexity. This motivates them to analyze the following parameterized problem: ``Given two finite relational structures and a parameter $k$, are the structures strongly $k$-consistent?'' They showed that this problem is complete for the parameterized complexity class co-W[2]. Hence, assuming FPT $\neq$ co-W[2], the problem is not solvable in time $O(f(k)n^c)$ for some computable function $f$ and constant $c$. We analyze the complexity of a stronger statement:  ``Given two finite relational structures and a parameter $k$, can strong $k$-consistency be established?'' The outcome of this decision problem matches the outcome of a $k$-consistency algorithm and thus characterizes the complexity of the $k$-consistency test precisely. The proof of Theorem \ref{maintheorem} implies that establishing strong $k$-consistency is complete for the parameterized complexity class XP. 
Thus, assuming co-W[2] $\neq$ XP, trying to make the instance $k$-consistent is a harder task than checking whether it is already $k$-consistent. 

\section{Pebble Games and Proof of Theorem \ref{maintheorem}} \label{sec:definitions}

\begin{figure}[htp]\centering
 \begin{tikzpicture}
  [scale=1, transform shape, knoten/.style={circle,draw=black,
  inner  sep=1pt,minimum  size=5mm},pebble/.style={circle,draw=black,fill=black,
  inner  sep=1pt,minimum  size=4mm}]
  
  \node[knoten] (u) at (0,0) {$u$};
  \node[knoten] (v) at (1,0) {$v$};
  \node[knoten] (w) at (2,0) {$w$};
  \node[pebble] (c) at (0,0.35) {\color{white} $c$};
  \node[pebble] (d) at (1,0.35) {\color{white} $d$};
  
  \node at (3.5,0) {$\Longrightarrow$};
  
  \node[knoten] (u1) at (5,0) {$u$};
  \node[knoten] (v1) at (6,0) {$v$};
  \node[knoten] (w1) at (7,0) {$w$};
  \node[pebble] (c1) at (7,0.35) {\color{white} $c$};
  \node[pebble] (d1) at (6,0.35) {\color{white} $d$};
  
\end{tikzpicture}
 \caption{KAI-game: Moving pebble $c$ according to rule $(u,v,w,c,d)$.}\label{fig:rule}
\end{figure}
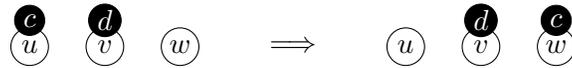

In this section we first introduce the KAI-game and the existential pebble game. Then we state the reduction from the KAI-game to the existential pebble game in our main lemma (Lemma \ref{lem:reduction}) and use it to prove Theorem \ref{maintheorem}.  

Kasai, Adachi and Iwata \cite{Kasai.1978} introduced a simple combinatorial pebble game that nicely simulates Turing machines. The authors showed in \cite{Kasai.1978} and \cite{Adachi.1984} that playing various variants of this game is complete for different complexity classes. Here we stick to the $k$-pebble variant, that is restricted to a fixed set $[k]:=\{1,\ldots,k\}$ of pebbles. 
An instance of the $k$-pebble KAI-game is a tuple $(X,R,\startpos,\goal)$, where $X$ is the set of nodes, $R=R'\times\{(c,d)\in [k]^2\mid c\neq d\}$ with $R'\subseteq [X]^3$ the set of rules, $\startpos\colon[k]\to X$ the start position and $\goal\in X$ the goal. A rule is of the form $(u,v,w,c,d)$, with distinct pebbles $c,d$, pairwise distinct nodes $u, v, w$ and the intended meaning that if pebble $c$ is on $u$ and pebble $d$ is on $v$ and there is no pebble on $w$ then one player can move pebble $c$ from $u$ to $w$ (see Figure \ref{fig:rule}). This is a slightly more wasteful notion than the original one used in \cite{Kasai.1978}, where the relation $R'\subseteq X^{3}$ (instead of $R'\times \{(c,d)\in [k]^2\mid c\neq d\}$) is given as input. However, this technical modification does not affect the purpose of the game and increases the size of an instance only by a constant factor if $k$ is fixed, and by a polynomial factor if $k$ is part of the input. A \textit{position} of the KAI-game is an injective mapping $\pos\colon[k]\to X$. A rule $r=(u,v,w,c,d)\in R$ is \textit{applicable} to a position $\pos$ if $\pos(c)=u$, $\pos(d)=v$ and $\pos(z)\neq w$ for all $z\in[k]$. Furthermore, if $r$ is applicable to $\pos$ then $r(\pos)$ denotes the position defined as $r(\pos)(c)=w$ and $r(\pos)(z) = \pos(z)$, for all $z\in[k]\setminus\{c\}$. The set of all rules in $R$ applicable to a position $\pos$ is denoted by $\appl(\pos)$, and 
$T_r(\pos)\subseteq [k]$ is the set of KAI-pebbles $i$ such that $\pos(i)$ contradicts the applicability condition of rule $r$:
$
T_{(u,v,w,c,d)}(\pos) \defi \{i\in[k]\mid (i=c\text{ and }\pos(i)\neq u)\text{ or }(i=d \text{ and } \pos(i)\neq v)\text{ or }\pos(i)=w\}.
$ Thus, $r\in \appl(\pos)$ iff $T_{r}(\pos)=\emptyset$.

The $k$-pebble KAI-game is played by two players and proceeds in rounds. In the first round Player 1 starts with position $\startpos$ and chooses a rule $r\in \appl(\startpos)$. The new position is $\pos = r(\startpos)$. In the next round Player 2 chooses a rule $r\in \appl(\pos)$ and applies it to $\pos$. Then it is Player 1's turn and so on. Player 1 wins the game if he reaches a position $\pos$, where $\pos(z)=\goal$ for one $z\in[k]$ or where Player 2 is unable to move. Player 2 wins if she has a strategy ensuring that Player 1 cannot reach such a position. The next definition formalizes winning strategies for Player 2. They contain sets of positions $\mathcal K_{i}, i\in\{1,2\}$ where it is Player i's turn and a mapping $\kappa$ that tells Player 2 for every position which rule to choose next.

\begin{definition}
A \textit{winning strategy} for Player 2 in the KAI-game on $(X,\{r_1,\ldots,r_m\},\startpos,\goal)$ is a triple $\mathcal K=(\mathcal K_1,\mathcal K_2, \kappa)$ where $\mathcal K_1\subseteq \{\pos\mid \pos\colon[k]\to X\}$ and $\mathcal K_2\subseteq \{\pos\mid \pos\colon[k]\to X\setminus\{\goal\}\}$ are sets of positions and $\kappa\colon\mathcal K_2\to [m]$ is a mapping such that the following holds:
\begin{itemize}
\item $\startpos \in \mathcal K_1$.
\item For every $\pos\in\mathcal K_1$ and every $r_i\in\appl(\pos)$: $r_i(\pos)\in\mathcal K_2$.
\item For every $\pos\in\mathcal K_2$: $r_{\kappa(\pos)}\in\appl(\pos)$ and $r_{\kappa(\pos)}(\pos)\in\mathcal K_1$.
\end{itemize}
\end{definition}

Kasai, Adachi and Iwata \cite{Kasai.1978} showed that the problem of determining the winner of the $k$-pebble KAI-game is PTIME-complete (for every fixed $k\geq 3$) under LOGSPACE-reductions and complete for EXPTIME when $k$ is part of the input. Furthermore, they proved the following unconditional lower bound.

\begin{theorem}[\cite{Adachi.1984}]\label{thm:KAI}
For every fixed $k \geq 6$ and any $\varepsilon > 0$, the winner of the $k$-pebble KAI-game on a given instance $I$ cannot be determined in time $O(\|I\|^{\frac{k-1}{4}-\varepsilon})$ on deterministic multi-tape Turing machines, where $\|I\|$ is the size of the input $I$.
\end{theorem} 

The proof of this theorem essentially relies on the deterministic time hierarchy theorem, which states that multi-tape Turing machines of running time $n^{k}$ cannot be simulated within time $n^{k-\epsilon}$. On the other hand, Turing machines of running time $n^{k}$ can be simulated within the $(4k+1)$-pebble KAI-game and hence the lower bound follows. In terms of parametrized complexity, their argument also leads to XP-completeness of the $k$-pebble KAI-game with parameter $k$ as pointed out in \cite{Downey.1999}. 

The existential ($k+1$)-pebble game \cite{Kolaitis.1995} is played by two players Spoiler and Duplicator on two relational structures \textbf{A} and \textbf{B} with domains A and B, respectively. First, Spoiler puts pebbles $a_1,\ldots,a_{k+1}$ on elements of A and Duplicator answers by putting pebbles $b_1,\ldots,b_{k+1}$ on elements of B. In each further round Spoiler picks up a pebble $a_i$ from A and places it on another element in A and Duplicator moves the corresponding pebble $b_i$ in B. Spoiler wins the game if he can reach a position where the mapping defined by $a_i\mapsto b_i$ is not a partial homomorphism from \textbf{A} to \textbf{B}. Duplicator wins the game if she has a winning strategy that tells her for every move of Spoiler how to place her pebbles such that the positions of the pebbles define a partial homomorphism. A winning strategy for Duplicator can be stated formally as a set of partial homomorphisms: 

\begin{definition}[\cite{Kolaitis.1995}] \label{def:win}
 A \textit{winning strategy} for Duplicator in the existential ($k+1$)-pebble game on structures \textbf{A} and \textbf{B} is a nonempty family $\mathcal H$ of partial homomorphisms from \textbf{A} to \textbf{B} satisfying the following properties:
\begin{description}
 \item[closure] If $h\in \mathcal H$ and $g\subset h$ then $g\in\mathcal H$.
 \item[extension] For every $g\in\mathcal H$, $|\dom(g)|\leq k$, and every $z\in A$ there is an $h\in\mathcal H$ with $g\subseteq h$ and $z\in\dom(h)$.
\end{description}
\end{definition}

\noindent For a set $H$ of partial homomorphisms from \textbf{A} to \textbf{B} we let $\cl(H) \defi \{g\mid g \subseteq h, h\in H\}$ be the closure of $H$ under taking subsets and write $\cl(h)$ instead of $\cl(\{h\})$. It is easy to see that if $h$ is a total homomorphism from \textbf{A} to \textbf{B}, then $\cl(h)$ is a winning strategy in the existential ($k+1$)-pebble game on \textbf{A} and \textbf{B}. 
Now we state our main lemma and prove Theorem \ref{maintheorem}. The proof of the main lemma is deferred to Section \ref{sec:mainlemma}.
\begin{lemma}[Main Lemma] \label{lem:reduction}
 There is a reduction from the $k$-pebble KAI-game to the existential ($k+1$)-pebble game that computes for every instance $I=(X,R,\startpos,\goal)$ two directed graphs $G_S$ and $G_D$ such that the following constraints hold:
\begin{itemize}
\item Player 1 has a winning strategy in the $k$-pebble KAI-game on $I$ if and only if Spoiler has a winning strategy in the existential ($k+1$)-pebble game on $G_S$ and $G_D$.
\item $|V(G_S)|+|V(G_D)| = O(|X|\cdot|R|\cdot k^2)$.
\item $|E(G_S)|+|E(G_D)|=O(k^{4}(|X|^{2}|R|+|X|\cdot|R|^{2}))$.
\item The reduction is computable in \DTIME$(O(\|I\|^{3}))$ and in \LOGSPACE. 
\end{itemize}
\end{lemma}

\begin{IEEEproof}[Proof of Theorem \ref{maintheorem}]
 Let $k\!\geq\! 15$ be a fixed integer and $\epsilon>0$. Assume that $\mathbb{A}$ is an algorithm that determines the winner of the existential $k$-pebble game on structures \textbf{A} and \textbf{B} in time $O((\|\textbf{A}\|\!+\!\|\textbf{B}\|)^{\frac{k-2}{12}-\epsilon})$. Let $\mathbb B$ be the algorithm that first applies the reduction from Lemma \ref{lem:reduction} to a given instance $I$ of the ($k-1$)-pebble KAI-game and then executes $\mathbb{A}$. Since $\|G_S\|\!+\!\|G_D\|=O(\|I\|^3)$, $\mathbb B$ has running time $O(\|I\|^3\!+\!\|I\|^{3\left(\frac{k-2}{12}-\epsilon\right)})$, and thus solves the $k'$-pebble KAI-game in time $O(\|I\|^{\frac{k'-1}{4}-\epsilon'})$ for $k'=k-1$ and $\epsilon'=\epsilon/3$. This contradicts Theorem \ref{thm:KAI}. \end{IEEEproof}
In addition, Lemma \ref{lem:reduction} also implies EXPTIME-completeness when $k$ is part of the input and PTIME-completeness for every fixed $k\geq 4$, hence subsumes the result of Kolaitis and Panttaja \cite{Kolaitis.2003}. Since the reduction is also an fpt-reduction, it follows that determining the winner in the existential $k$-pebble game is complete for the parameterized complexity class XP.

In our reduction we first construct two colored graphs out of smaller graphs, called gadgets. In order to prove the existence of a winning strategy for one player, we combine strategies for the gadgets to a strategy for the whole graph. The easier part is to do that for Spoiler. As in \cite{Grohe.1996} and \cite{Kolaitis.2003}, we say that Spoiler \textit{can reach} position $p_j$ from position $p_i$ if he has a strategy in the game such that starting from position $p_i$ he wins the game, or position $p_j$ occurs in the game after some finite number of rounds. Since this relation is transitive, we can combine such strategies to show that Spoiler can reach some position $p$ from $\emptyset$; if $p$ does not define a partial homomorphism, this gives us a winning strategy for Spoiler.

For Duplicator this is more difficult. A \textit{critical strategy} in the existential ($k+1$)-pebble game is a nonempty family $\mathcal H$ of partial homomorphisms satisfying the closure property (Definition \ref{def:win}) together with a set of \textit{critical positions} $\crit(\mathcal H)\subset \mathcal H$ such that $h\in \crit(\mathcal H) \Longrightarrow |\dom(h)|=k$ and all $g\in \mathcal H \setminus \crit(\mathcal H)$ satisfy the extension property. A critical strategy is nearly a winning strategy in the sense that Duplicator wins unless the game reaches a critical position. Note that a critical strategy with $\crit(\mathcal H)=\emptyset$ is a winning strategy and every critical strategy in the ($k+1$)-pebble game is a winning strategy in the $k$-pebble game. Let $\mathcal{\hat{H}} \defi \mathcal H \setminus \crit(\mathcal H)$. As for winning strategies, the union of critical strategies is also a critical strategy. The following lemma enables us to construct a winning strategy out of critical strategies.
\begin{lemma}\label{lem:combunion}
If $\mathcal H_1,\ldots,\mathcal H_l$ are critical strategies on the same structures and for all $i\in[l]$ and all $p\in\crit(\mathcal H_i)$ there exists a $j\in[l]$ such that $p\in \mathcal{\hat{H}}_j$, then $\bigcup_{i\in[l]}\mathcal H_i$ is a winning strategy on these structures.  \qed
\end{lemma}

Every gadget $Q$ that we construct consists of two graphs $Q_S$ and $Q_D$ for Spoiler's and Duplicator's side, respectively. The graphs contain \textit{boundary vertices} $\bd(Q_S)\subseteq V(Q_S)$ and $\bd(Q_D)\subseteq V(Q_D)$, which are the vertices shared with other gadgets. That is, vertices in $V(Q_S)\setminus\bd(Q_S)$ ($V(Q_D)\setminus\bd(Q_D)$) are only adjacent to vertices in $V(Q_S)$ ($V(Q_D)$). A \textit{boundary function} of a strategy $\mathcal{H}$ on a gadget $Q$ is a mapping $\beta\colon\bd(Q_S)\to\bd(Q_D)$ such that $\beta(z)=h(z)$ for all $h\in\mathcal H$ and all $z\in \bd(Q_S)\cap\dom(h)$. We say that two strategies $\mathcal G$ and $\mathcal H$ on gadgets $Q$ and $Q'$ are \textit{connectable}, if they have boundary functions $\beta_{\mathcal G}$ and $\beta_{\mathcal H}$ and it holds that $\beta_{\mathcal G}(z) = \beta_{\mathcal H}(z)$ for all $z\in\bd(Q_S)\cap\bd(Q'_S)$.  If $\mathcal G$ and $\mathcal H$ are two connectable strategies, we define the \textit{composition}
$$
\mathcal G \circ \mathcal H = \{g\cup h\mid g\in\mathcal G, h\in\mathcal H\}.
$$
\begin{lemma} \label{lem:combcirc}
 Let $\mathcal G$  and $\mathcal H$ be two connectable critical strategies on gadgets $Q=(Q_S,Q_D)$ and $Q'=(Q'_S,Q'_D)$, respectively. The composition $\mathcal G\circ\mathcal H$ is a critical strategy on $Q_S\cup Q'_S$ and $Q_D\cup Q'_D$ with $\crit(\mathcal G\circ\mathcal H) = \crit(\mathcal G)\cup\crit(\mathcal H)$. \qed
\end{lemma}
Playing according to the strategy $\mathcal G \circ \mathcal H$ on $Q$ and $Q'$ means that Duplicator uses strategy $\mathcal G$ on $Q$ and strategy $\mathcal H$ on $Q'$. The requirements on the boundary ensure that strategy $\mathcal G$ equals strategy $\mathcal H$ on the intersection of $Q$ and $Q'$. 
We use the operator $\circ$ to construct global critical strategies for the whole graph out of critical strategies on the gadgets. Then we show that the union of those global critical strategies is by Lemma \ref{lem:combunion} a winning strategy for Duplicator.

\section{The Reduction}
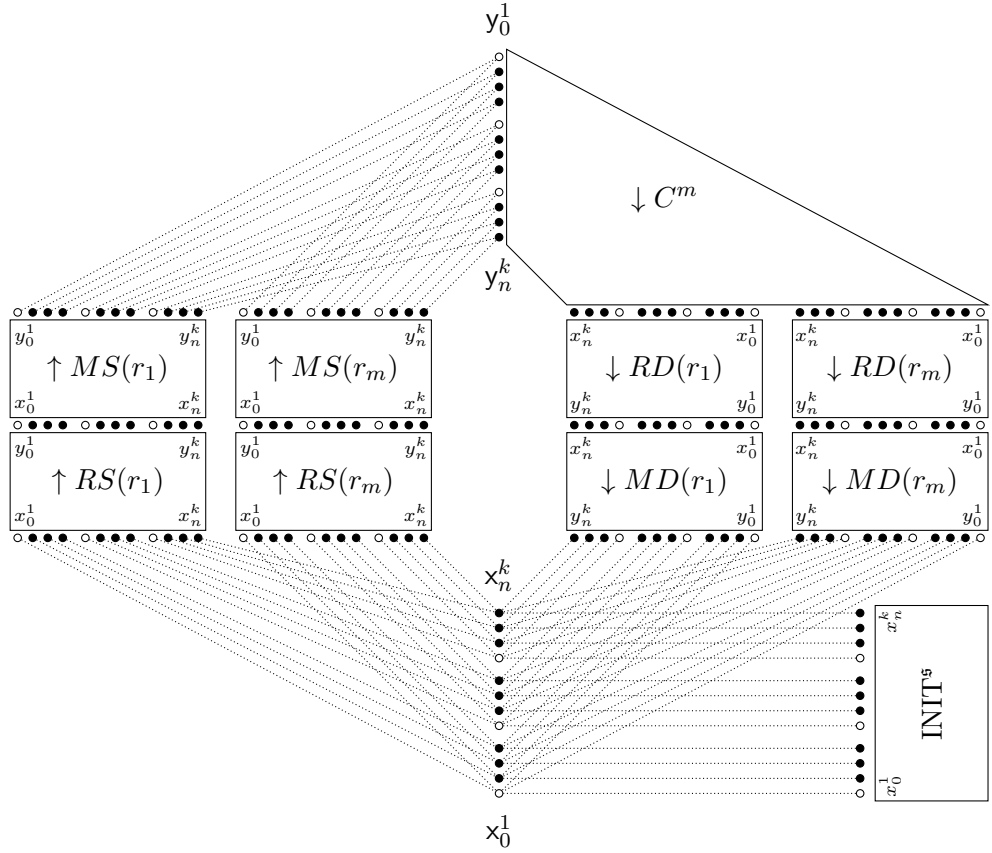
\begin{figure*}[htp]
 \centering

\input{Tikz/Tikz-Glueing-LICS.tex}
 \caption{The graph $G_D$. The dotted edges are only for visual purposes and need to be contracted. The gadgets $MS(r_i)$ and $MD(r_i)$ are distinct copies of the switch $M^{k,n}$. The special vertices $\xnode^i_j$ and $\ynode^i_j$ are also depicted in Figure \ref{fig:blocks}.}
 \label{fig:constr}
\end{figure*}

Let $([n],R,\startpos,\goal)$ be an instance of the $k$-pebble KAI-game and $m\defi |R|$. As in \cite{Kolaitis.2003}, the main idea is to simulate every play of the KAI-game within the existential pebble game such that Spoiler imitates the moves of Player 1 and Duplicator imitates the moves of Player 2. 
First, we construct two colored simple graphs, $G_S$ and $G_D$, and then show how to omit the colors while switching to directed graphs. We use $|V(G_S)|$ colors to color every vertex of Spoiler's graph $G_S$ differently and partition the vertices of Duplicator's graph with these colors. Thus, whenever Spoiler pebbles a vertex in $G_S$ there is a corresponding set of vertices in $G_D$ Duplicator can pebble.

To encode a position of the KAI-game in the existential ($k+1$)-pebble game we introduce the vertices $\{\xnode^{1},\ldots,\xnode^{k}\}$ in Spoiler's graph and $\{\xnode^{i}_l\mid i\in [k],0\leq l\leq n\}$ in Duplicator's graph. 
For each $i$, all the vertices $\{\xnode^{i}\}\cup\{\xnode^{i}_0,\ldots,\xnode^{i}_n\}$ are colored with the same unique color, denoted by $c_{\xnode^{i}}$.
The vertices $\xnode^{i}_0$ play a special role in the construction, so we draw \whitenode for vertices with subindex $0$ in the figures and \blacknode for all other vertices. Furthermore, we introduce vertices $\ynode^{i}$,$\ynode^{i}_l$ in the same way.

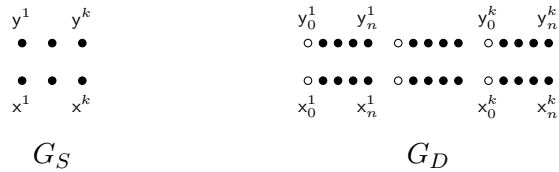
\begin{figure*}[htp]
\begin{center}
\begin{tikzpicture}
      [scale=1, transform shape, knoten/.style={circle,draw=black,fill=black,
      inner  sep=0pt,minimum  size=1mm},kknoten/.style={circle,draw=black,fill=black,
      inner  sep=0pt,minimum  size=0.8mm}, leerknoten/.style={circle,draw=black,fill=white,
      inner  sep=0pt,minimum  size=1mm},leerkknoten/.style={circle,draw=black,fill=white,
      inner  sep=0pt,minimum  size=0.8mm}]
  
  \node[knoten] (x1) at (-1.8-2,0) [label=below:{\tiny{$\xnode^{1}$}}] {};
  \node[knoten] (x2) at (-1.4-2,0)  {};
  \node[knoten] (x3) at (-1-2,0) [label=below:{\tiny{$\xnode^{k}$}}] {};
  
  \node[leerknoten] (x10) at (0,0) [label=below:{\tiny{$\xnode^{1}_0$}}] {};
  \node[knoten] (x11) at (0.2,0)  {};
  \node[knoten] (x12) at (0.4,0)  {};
  \node[knoten] (x13) at (0.6,0) {};
  \node[knoten] (x14) at (0.8,0) [label=below:{\tiny{$\xnode^{1}_n$}}] {};

  \node[leerknoten] (x20) at (1.2,0) {};
  \node[knoten] (x21) at (1.4,0)  {};
  \node[knoten] (x22) at (1.6,0)  {};
  \node[knoten] (x23) at (1.8,0)  {};
  \node[knoten] (x24) at (2,0)  {};

  \node[leerknoten] (x30) at (2.4,0) [label=below:{\tiny{$\xnode^{k}_0$}}] {};
  \node[knoten] (x31) at (2.6,0){};
  \node[knoten] (x32) at (2.8,0)  {};
  \node[knoten] (x33) at (3,0)  {};
  \node[knoten] (x34) at (3.2,0) [label=below:{\tiny{$\xnode^{k}_n$}}] {};
  
\begin{scope}[yshift=0.5cm]
  \node[knoten] (x1) at (-1.8-2,0) [label=above:{\tiny{$\ynode^{1}$}}] {};
  \node[knoten] (x2) at (-1.4-2,0)  {};
  \node[knoten] (x3) at (-1-2,0) [label=above:{\tiny{$\ynode^{k}$}}] {};
  
  \node[leerknoten] (x10) at (0,0) [label=above:{\tiny{$\ynode^{1}_0$}}] {};
  \node[knoten] (x11) at (0.2,0)  {};
  \node[knoten] (x12) at (0.4,0)  {};
  \node[knoten] (x13) at (0.6,0) {};
  \node[knoten] (x14) at (0.8,0) [label=above:{\tiny{$\ynode^{1}_n$}}] {};

  \node[leerknoten] (x20) at (1.2,0) {};
  \node[knoten] (x21) at (1.4,0)  {};
  \node[knoten] (x22) at (1.6,0)  {};
  \node[knoten] (x23) at (1.8,0)  {};
  \node[knoten] (x24) at (2,0)  {};

  \node[leerknoten] (x30) at (2.4,0) [label=above:{\tiny{$\ynode^{k}_0$}}] {};
  \node[knoten] (x31) at (2.6,0){};
  \node[knoten] (x32) at (2.8,0)  {};
  \node[knoten] (x33) at (3,0)  {};
  \node[knoten] (x34) at (3.2,0) [label=above:{\tiny{$\ynode^{k}_n$}}] {};
\end{scope}
\node at (-3.4,-1){$G_S$};
\node at (1.6,-1) {$G_D$};

\end{tikzpicture}
\end{center}
\caption{Vertex blocks to encode positions in the KAI-game.} \label{fig:blocks}
\end{figure*}

If $\pos\colon[k]\to[n]$ is a position of the $k$ pebbles in the KAI-game and it is Player 1's turn, then $\{(\xnode^{i},\xnode^{i}_{\pos(i)})\mid i\in [k]\}$ is the corresponding position in the existential pebble game. If it is Player 2's turn, then $\{(\ynode^{i},\ynode^{i}_{\pos(i)})\mid i\in [k]\}$ is the corresponding position. During the course of the game Spoiler pebbles some vertex $\xnode^i$ asking, ``Where does KAI-pebble $i$ lie?'' Due to the coloring Duplicator has to answer with some vertex $\xnode^i_l$ meaning ``KAI-pebble $i$ lies on node $l$.''
The vertices \whitenode are used to handle the case when Spoiler does not play in the intended way, that is, Spoiler has a winning strategy if and only if he has a winning strategy on the \blacknode vertices. In order to name positions that include \whitenode vertices, we define for positions $\pos$ and sets $T\subseteq [k]$ the mapping $(\pos,T)$ as 
$$
(\pos,T)(i) = \begin{cases} 0\text{, }i\in T, \\ \pos(i)\text{, else,}\end{cases}
$$
and write $\pos$ for $(\pos,\emptyset)$ and $\boldsymbol{0}$ for $(\pos,[k])$. 
 Now we have to introduce gadgets to ensure that Spoiler can simulate a play of the KAI-game. That is, if Player 1 can reach a position $\pos$ in the KAI-game, then Spoiler can reach the encoded position on the $\xnode$- or $\ynode$-vertices. The following list of properties ensures this.

\begin{itemize}
\item For the start position $\startpos$, Spoiler can reach $\{(\xnode^{i},\xnode^{i}_{\startpos(i)})\mid i\in [k]\}$ from $\emptyset$.
\item For a position $\pos$ and every rule $r\in\appl(\pos)$, Spoiler can reach $\{(\ynode^{i},\ynode^{i}_{r(\pos)(i)})\mid i\in [k]\}$ from $\{(\xnode^{i},\xnode^{i}_{\pos(i)})\mid i\in [k]\}$.
\item Spoiler can reach $\{(\xnode^{i},\xnode^{i}_{r(\pos)(i)})\mid i\in [k]\}$ from $\{(\ynode^{i},\ynode^{i}_{\pos(i)})\mid i\in [k]\}$ for a rule $r\in\appl(\pos)$ of Duplicator's choice.
\item $\{(\ynode^{i},\ynode^{i}_{\pos(i)})\mid i\in [k]\}$ is not a partial homomorphism if $\pos(i)=\goal$ for some $i\in[k]$.
\end{itemize}

It follows from these properties that if Player 1 has a winning strategy in the KAI-game, then Spoiler wins the existential pebble game by simulating Player 1's winning strategy. The difficult task is to prove that this is the only way for Spoiler to win.   
We give a brief description of the construction and argue how Spoiler is intended to play on it. Duplicator's graph is illustrated in Figure \ref{fig:constr}. The gadgets are glued together at their boundary vertices and the vertex blocks that are glued together inherit their colors. In order to make sure that the colors partition the graphs we define a new color for every combination of colors occurring at one vertex in the graph. In Spoiler's graph we proceed the same way with Spoiler's side of the gadgets.

To implement the last condition, we simply delete the color $c_{\ynode^{i}}$ from $\ynode^{i}_\goal$ for all $i\in[k]$. Since $\ynode^{i}$ and all $\ynode^{i}_l$, $l\in[n]\setminus\{\goal\}$, are still colored $c_{\ynode^{i}}$, it follows that the mapping $\ynode^{i}\mapsto\ynode^{i}_{l}$ that encodes ``KAI-pebble $i$ lies on node $l$'' is a partial homomorphism if and only if $l$ is not the goal node $\goal$. It follows that Spoiler wins the game if he can reach the position $\{(\ynode^{i},\ynode^{i}_{\pos(i)})\mid i\in [k]\}$ where $\pos$ is a winning position for Player 1 in the KAI-game. 

To make sure that Spoiler can reach the start position, we introduce the initialization gadget INIT$^{\startpos}$, whose boundary $x^{1},\ldots,x^{k}$ in Spoiler's graph and $x^{1}_0,\ldots,x^{k}_n$ in Duplicator's graph is identified with vertices $\xnode^{1},\ldots,\xnode^{k}$ in Spoiler's graph and $\xnode^{1}_0,\ldots,\xnode^{k}_n$ in Duplicator's graph. 
 The boundary vertices of the other gadgets have a similar form and can be divided into input vertices $x$ (with certain indices) and output vertices $y$ that are colored in the same way as the $\xnode$- and $\ynode$-vertices. As above, a position $\pos$ in the KAI-game is encoded as $\{(x^i,x^i_{\pos})\mid i\in [k]\}$ on these vertex blocks and we call it \emph{$\pos$ on $x$}.
The direction of the gadgets is indicated in Figure \ref{fig:constr} by arrows. Thus, the players are intended to move clockwise in the graph. 

For each rule $r$ we define different rule gadgets $RS(r)$ and $RD(r)$ in which Spoiler can reach the position $r(\pos)$ on the output $y$ from $\pos$ on the input $x$ if $r$ is applicable to $\pos$. Hence, from a position $\pos$ on $\xnode$ Spoiler can choose an applicable rule $r$ and reach $r(\pos)$ on the output $y$ of some rule gadget $RS(r)$. The choice gadget $C^{m}$ enables Duplicator to choose one of the $m$ rules she wants to apply. That is, Duplicator can choose a rule $r$ such that from $\pos$ on $\ynode$ Spoiler can reach $\pos$ on the input of $RD(r)$ and then $r(\pos)$ on the output of $RD(r)$. 
The most complex gadget is the multiple input one-way switch $M^{k,n}$, which is a generalization of the multiple input one-way switch defined in \cite{Kolaitis.2003}. 
In our construction we put one copy of $M^{k,n}$ at the output vertices of every rule gadget. Spoiler's strategy on this gadget is nevertheless simple: he can pebble a position through the switch, that is, he can reach $\pos$ on the output from $\pos$ on the input. This concludes the description of how the gadgets can be used by Spoiler to ensure the four properties above. 

Duplicator's strategy is to force Spoiler to play exactly this way. Especially, if the KAI-game does not stop, then Duplicator can play the existential pebble game forever by forcing Spoiler to simulate this infinite play. The main tool for Duplicator is to answer with \whitenode vertices whenever Spoiler plays incorrectly: if Spoiler pebbles a vertex $x^i$ he is not supposed to pebble now, then Duplicator answers with $x^i_0$.  The strategies on the gadgets ensure that such positions extend to partial homomorphisms and thus Spoiler does not benefit from them.

\subsection{Rule Gadgets}

The rule gadgets $RS(r)$ and $RD(r)$ consist of input vertices $x^1,\ldots , x^k$ in Spoiler's graph and $x^1_0,\ldots,x^k_n$ in Duplicator's graph, and output vertices $y^1,\ldots, y^k$ and $y^1_0,\ldots,y^n_k$. For each rule $r=(u,v,w,c,d)$ we connect the vertices in the gadgets $RS(r)$ and $RD(r)$ as shown in Figure \ref{fig:RS} and \ref{fig:RD}. 

\begin{figure}
 \centering
\input{Tikz/Tikz-RS-LICS.tex}
 
  \caption{Rule gadget $RS(u,v,w,c,d)$. The range of $i$ is $[k]\setminus\{c,d\}$.}\label{fig:RS}
\end{figure}
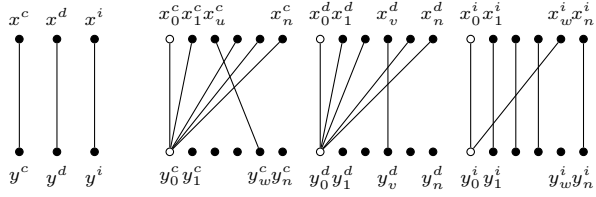

\begin{figure}
 \centering
\input{Tikz/Tikz-RD-LICS.tex}
 
 \caption{Rule gadget $RD(u,v,w,c,d)$. The range of $i$ is $[k]\setminus\{c,d\}$.}\label{fig:RD}
\end{figure}
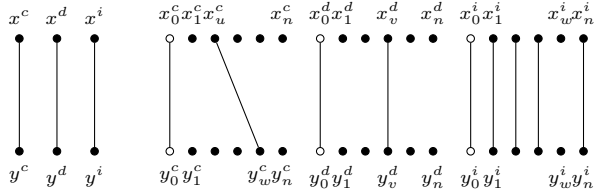

If $r$ is applicable to $\pos$, then Spoiler can reach the position $\{(y^i,y^i_{r(\pos)(i)})\mid i\in[k]\}$ from $\{(x^i,x^i_{\pos(i)})\mid i\in [k]\}$ in both gadgets. To do so, he picks up the remaining pebble and puts it on $y^{1}$. Then he picks up the pebble from $x^{1}$ and puts it on $y^{2}$ and so on. This fact is stated in Lemma \ref{lem:srule}(i) and \ref{lem:drule}(i). Assume that $r$ is not applicable to $\pos$ (then $T_r(\pos)\neq\emptyset$) and the current position is $\{(x^i,x^i_{\pos(i)})\mid i\in [k]\}$. On $RS(r)$ Duplicator can pebble some \whitenode vertex $y^i_0$  when Spoiler asks for $y^i$ ($i\in T_r(\pos)$), and thus can avoid valid positions on the $y$-vertices. Therefore, Spoiler is penalized when he chooses a rule $r$ not applicable to $\pos$ and plays on $RS(r)$. This strategy is stated in Lemma \ref{lem:srule}(ii) for $T=\emptyset$. If Duplicator chooses a rule $r$ not applicable to $\pos$ and plays on $RD(r)$, then she will be penalized, because Spoiler wins immediately from position $\{(x^i,x^i_{\pos(i)})\mid i\in [k]\}$ (Lemma \ref{lem:drule}(iii)) by pebbling on $y^{i}$ for some $i\in T_r(\pos)$. Furthermore, Lemma \ref{lem:srule}(ii) and \ref{lem:drule}(ii) state that for invalid positions  on the $x$-vertices (i.e. $T\neq \emptyset$), Duplicator can avoid valid positions on the $y$-vertices. 

\begin{lemma} \label{lem:srule}
 For every rule $r=(u,v,w,c,d)$ and position $\pos\colon[k]\to [n]$ the following holds in the existential ($k+1$)-pebble game on $RS(r)$:
\begin{enumerate}[leftmargin=2.2 em,label=(\roman*)]
 \item If $r\in\appl(\pos)$, then Spoiler can reach $\{(y^i,y^i_{r(\pos)(i)})\mid i\in[k]\}$ from $\{(x^i,x^i_{\pos(i)})\mid i\in [k]\}$.
 \item  Duplicator has a winning strategy $\mathcal R_{(\pos,T)}$ with boundary function $\{(x^i,x^i_{(\pos,T)(i)})\mid i\in[k]\}\cup \{(y^i,y^i_{(r(\pos),T\cup T_r(\pos))(i)})\mid i\in[k]\}$, for all $T\subseteq [k]$.
\end{enumerate}
\end{lemma}

\begin{proof}
  For $i=1,\ldots,k$ Spoiler takes the remaining pebble and puts it on $y^i$. Since there is an edge $\{x^i,y^i\}$, Duplicator has to answer with $y^i_{r(\pos)(i)}$ because this is the only vertex adjacent to $x^i_{\pos(i)}$. In the next step Spoiler picks up the pebble pair from $x^i, x^i_{\pos}$ and proceeds with $i+1$.

  The boundary function defined in (ii) preserves the vertex colors and maps edges $\{x^i,y^i\}$ to edges $\{x^i_{(\pos,T)(i)}, y^i_{(r(\pos),T\cup T_r(\pos))(i)}\}$, hence defines a total homomorphism on $RS(r)$. It follows that $\mathcal R_{(\pos,T)}\defi \{h\mid h\subseteq \beta\}$ is a winning strategy.
\end{proof}

\begin{lemma}\label{lem:drule}
 For every rule $r=(u,v,w,c,d)$ and position $\pos\colon[k]\to [n]$ the following holds in the existential ($k+1$)-pebble game on $RD(r)$:
\begin{enumerate}[leftmargin=2.2 em,label=(\roman*)]
 \item If $r\in\appl(\pos)$, then Spoiler can reach $\{(y^i,y^i_{r(\pos)(i)})\mid i\in[k]\}$ from $\{(x^i,x^i_{\pos(i)})\mid i\in [k]\}$.
 \item  If $r\in\appl(\pos)$, then Duplicator has a winning strategy $\mathcal R_{(\pos,T)}$ with boundary function $\{(x^i,x^i_{(\pos,T)(i)})\mid i\in[k]\}\cup \{(y^i,y^i_{(r(\pos),T)(i)})\mid i\in[k]\}$, for all $T\subseteq [k]$.
 \item If $r\notin\appl(\pos)$, then Spoiler wins from $\{(x^i,x^i_{\pos(i)})\mid i\in [k]\}$.
\end{enumerate}
\end{lemma}

\begin{proof}
  Statement (i) is analog to Lemma \ref{lem:srule}(i). The boundary function $\beta$ stated in (ii) defines a total homomorphism on $RD(r)$. Thus, $\mathcal R_{(\pos,T)}\defi \{h\mid h\subseteq \beta\}$ is a winning strategy.

  In order to win from $\{(x^i,x^i_{\pos(i)})\mid i\in [k]\}$ for a rule $r$ not applicable to $\pos$ (Statement (iii)), Spoiler chooses some KAI-pebble $i\in T_r(\pos)$ (whose position $\pos(i)$ witnesses that $r$ is not applicable to $\pos$) and puts the remaining pebble on $y^i$. Since Duplicator has to answer with a vertex $y^i_j$ of the same color and none of them is adjacent to $x^i_{\pos(i)}$, she looses immediately.
\end{proof}

\subsection{The Multiple Input One-Way Switch}

As for the rule gadgets, the multiple input one-way switch has input vertices $x^1,\ldots,x^k$ in Spoiler's graph and $x^1_0,\ldots,x^k_n$ in Duplicator's graph, and output vertices $y^1,\ldots,y^k$ and $y^1_0,\ldots,y^k_n$, respectively. We say that a position $(\pos,T)$ is on the input (output) to denote the positions $\{(x^i,x^i_{(\pos,T)(i)})\mid i\in [k]\}$ ($\{(y^i,y^i_{(\pos,T)(i)})\mid i\in [k]\}$). 

For Spoiler the switch ensures that he can reach $\pos$ on the output from $\pos$ on the input (Lemma \ref{lem:multi}(i)). 
For Duplicator there are several strategies. She has a winning strategy called \textit{output strategy}, where any position is on the output and $\boldsymbol{0}$ is on the input (Lemma \ref{lem:multi}(ii)). This ensures that Spoiler cannot move backwards and reach $\pos$ on the input from $\pos$ on the output. Next, for every nonempty $T\subseteq [k]$ Duplicator has a winning strategy where $(\pos,T)$ is on the input and $\boldsymbol{0}$ is on the output (Lemma \ref{lem:multi}(iii)). Thus, she has a strategy such that Spoiler can reach only the $\boldsymbol{0}$ position on the output from invalid positions on the input. These strategies are called \textit{restart strategies}. We will see later that Spoiler has to restart the game, that is, he has to pick up all pebbles and start playing on the initialization gadget, if he reaches a position that is contained in a restart strategy. To ensure that Spoiler picks up all pebbles when reaching $\pos$ on the output from $\pos$ on the input, Duplicator has a critical \textit{input strategy} with $\pos$ on the input and $\boldsymbol{0}$ on the output, whose critical positions are either contained in an output strategy (where $\pos$ is on the output) or in a restart strategy (Lemma \ref{lem:multi}(iv)). 
If Duplicator plays according to this input strategy, the only way for Spoiler to bring $\pos$ from the input to the output is to pebble a critical position inside the switch (using all the pebbles) and force Duplicator to switch to the corresponding output strategy.

\begin{figure*}[htp]
 \centering
  \resizebox{\textwidth}{!}{
  \input{Tikz/Tikz-multi-LICS.tex}
 }
 \caption{Subgraph of the multiple input one-way switch $M^{k,n}$. For the first block in Duplicator's side, all inner-block edges are drawn. Note that there is no edge between $a^i_{s,l}$ and $b^i_{0,l}$.}\label{fig:multi}
\end{figure*}
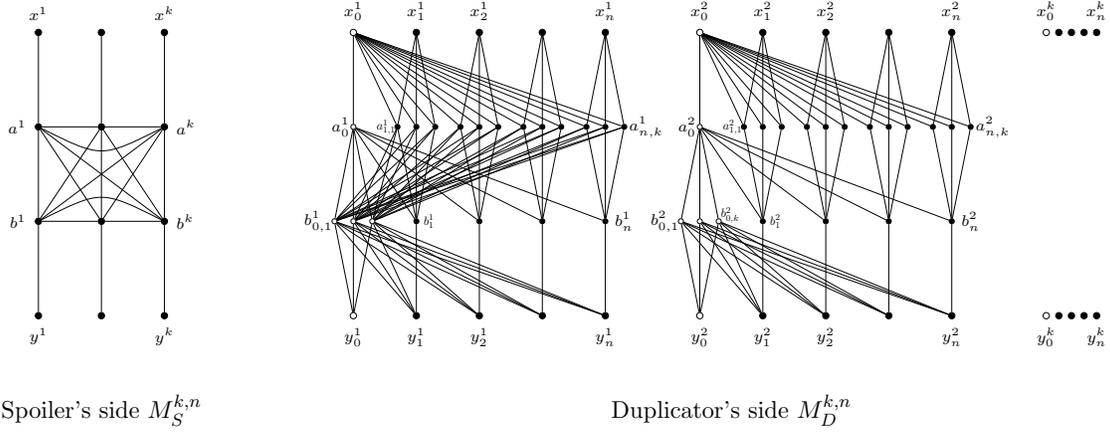

In order to define the \textit{multiple input one-way switch} $M^{k,n}$ we construct the two graphs $M_S^{k,n}$ for Spoiler's side and $M_D^{k,n}$ for Duplicator's side.
Let
\begin{align*}
  V(M^{k,n}_S) &= \{x^i,a^i,b^i,y^i\mid i\in [k]\}, \\
  E(M^{k,n}_S) &= \big\{\{x^i,a^i\},\{b^i,y^i\}\mid i\in[k]\big\} \\
  &\qquad\cup \big\{\{a^{i},a^{j}\},\{b^{i},b^{j}\}\mid i,j\in[k]; i\neq j\big\} \\
  &\qquad\cup \big\{\{a^{i},b^{j}\}\mid i,j\in[k]\big\}.
\end{align*}
That is, $M^{k,n}_S$ simply consists of $k$ input vertices $x^{i}$ and $k$ output-vertices $y^{i}$, which are each connected to one vertex of a $2k$-clique.
For Duplicator's side of the graph, we define for $i\in [k]$:
\begin{align*}
 X^i &= \{x^i_s\mid 0\leq s\leq n\},& Y^i &= \{y^i_s\mid 0\leq s\leq n\}\\
 A^i_+ &= \{a^i_{s,l}\mid s\in [n],l\in [k]\}, 
 &A^i &= A^i_+ \cup \{a^i_0\} \\
 B^i_+ &= \{b^i_s \mid s \in [n]\},
 &B^i &= B^i_+ \cup \{b^i_{0,l}\mid l\in [k]\}. 
\end{align*}
 The set of vertices of $M_D^{k,n}$ is 
$$
V(M_D^{k,n}) = \bigcup_{i\in [k]}\left(X^i\cup A^i \cup B^i \cup Y^i\right).
$$
The graph consists of $k$ blocks, where the $i$-th block contains the vertices $X^i\cup A^i \cup B^i \cup Y^i$. For every $i\in[k]$ we color $\{x^i\}\cup X^i$ as well as $\{a^i\}\cup A^i$, $\{b^i\}\cup B^i$ and $\{y^i\}\cup Y^i$ with a unique color. We first define the inner-block edges $E^i$, which are also shown in Figure \ref{fig:multi}, and then the inter-block edges $E^{i,j}$ (for notational convenience we always assume $l,p\in[k]$ and $s,q\in[n]$):
\begin{align*}
 E^i &=  &&\big(\{x^i_0\}\times A^i\big) \cup 
  &&\big\{ \{x^i_s, a^i_{s,l} \}\big\} \cup
  &&\big(\{a_0^i\}\times B^i\big) \cup 
  &&\text{(E1)-(E3)} \\
  &&&\big\{ \{a^i_{s,l}, b^i_{s}\}\big\} \cup
  &&\big\{\{a^i_{s,l},b^i_{0,p}\}\mid l\neq p\big\} \cup
  &&\big\{ \{b^i_s,y^i_s\}\big\} \cup 
  &&\text{(E4)-(E6)} \\
  &&&\big(\{b^i_{0,l}\mid l\in[k]\}\times Y^i\big), 
  &&&&&&\text{(E7)} \\
E^{i,j} &= &&\big\{\{a^i_{s,l},a^j_{q,p}\}\mid l\neq p \big\} \cup
  &&\big\{\{a^i_{0},a^j_{s,l}\} \big\}\cup
  &&\big\{\{b^i_{0,l},b^j_{0,p}\} \big\}\cup
  &&\text{(E8)-(E10)} \\
  &&&\big\{\{b^i_{s},b^j_{q}\}\big\}\cup
  &&\big\{\{a^i_{0},b^j_{0,l}\}\big\}\cup
  &&\big\{\{a^i_{0},b^j_{s}\}  \big\}\cup
  &&\text{(E11)-(E13)} \\
  &&&\big\{\{a^i_{s,l},b^j_{q}\}   \big\}\cup
  &&\big\{\{a^i_{s,l},b^j_{0,p}\}\mid l\neq p\big\}.
  &&&&\text{(E14)-(E15)}
\end{align*}
Finally, $E(M^{k,n}_D) = \bigcup_{i\in[k]}E^i\cup \bigcup_{i,j\in[k];i\neq j}E^{i,j}$.
Our switch is a further development of the switch used by Kolaitis and Panttaja \cite{Kolaitis.2003}. In fact, their
multiple input one-way switch $M^{k+1}$ is isomorphic to $M^{k,2}$. 
%
The next lemma states the main properties of the switch.
\begin{lemma} \label{lem:multi}
 For every position $\pos\colon[k]\to[n]$, the following statements hold in the existential ($k+1$)-pebble game on $M^{k,n}$:
\begin{enumerate}[leftmargin=2.2 em,label=(\roman*)]
 \item Spoiler can reach $\{(y^{1},y^1_{\pos(1)}),\ldots,(y^{k},y^k_{\pos(k)})\}$ from $\{(x^{1},x^1_{\pos(1)}),\ldots,(x^{k},x^k_{\pos(k)})\}$.
 \item Duplicator has a winning strategy $\mathcal H^\text{out}_{(\pos,T)}$ with boundary function $\{(x^i,x^i_0)\mid i\in [k]\}$ and 
 $\{(y^i,y^i_{(\pos,T)(i)})\mid i\in [k]\}$ for all $T\subseteq [k]$.
 \item Duplicator has a winning strategy $\mathcal H^\text{restart}_{(\pos,T)}$ with boundary function $\{(x^i,x^i_{(\pos,T)(i)})\mid i\in [k]\}$ and $\{(y^i,y^i_0)\mid i\in [k]\}$ for all nonempty $T\subseteq [k]$.
 \item Duplicator has a critical strategy $\mathcal H^\text{in}_{\pos}$ with boundary function $\{(x^i,x^i_{\pos(i)})\mid i\in [k]\}\cup \{(y^i,y^i_0)\mid i\in [k]\}$ and sets of critical positions $\mathcal C^\text{restart-crit}_{\pos,t}$ (for $t\in [k]$) and $\mathcal C^\text{out-crit}_{\pos}$ such that:
  \begin{enumerate}[label=\({alph*}]
    \item $\crit(\mathcal H^\text{in}_{\pos}) = \bigcup_{t\in[k]} \mathcal C^\text{restart-crit}_{\pos,t}\cup \mathcal C^\text{out-crit}_{\pos}$,
    \item $\mathcal C^\text{restart-crit}_{\pos,t} \subseteq \mathcal H^\text{restart}_{(\pos,\{t\})}$ and
    \item $\mathcal C^\text{out-crit}_{\pos} \subseteq \mathcal H^\text{out}_{\pos}$.
  \end{enumerate}
\end{enumerate}
\end{lemma}

\begin{IEEEproof}
Let $\pos\colon[k]\to[n]$ be an arbitrary position. We first construct the strategy for Spoiler to prove (i). Starting from position $\{(x^{1},x^1_{\pos(1)}),\ldots,(x^{k},x^k_{\pos(k)})\}$, Spoiler places the ($k+1$)st pebble on $a^1$. Duplicator has to answer with $a^1_{\pos(1),l_1}$ for some $l_1\in[k]$, mapping the edge $\{x^1,a^1\}$ to some edge in (E2). Next, Spoiler picks up the pebble from $x^1$ and puts it on $a^2$. Again, Duplicator has to answer with $a^2_{\pos(2),l_2}$ for some $l_2\in [k]\setminus\{l_1\}$. The index $l_2$ has to be different from $l_1$ because there is an edge between $a_1$ and $a_2$, but none between $a^1_{\pos(1),l_1}$ and $a^2_{\pos(2),l_1}$ in (E8). Following that scheme, Spoiler can reach the position $\{(a^{1},a^1_{\pos(1),l_1}),\ldots,(a^{k},a^k_{\pos(k),l_k})\}$ for pairwise distinct $l_1, l_2, \cdots, l_k$. Now, Spoiler pebbles $b^1$ with the free pebble and Duplicator has to answer with a vertex in $B^1$ (due to the vertex-colors) that is adjacent to all $a^1_{\pos(1),l_1},\ldots,a^k_{\pos(k),l_k}$. This is only the case for $b^1_{\pos(1)}$ (due to (E4) and (E14)), since every vertex of the form $b^1_{0,l_i}$ is not adjacent to the vertex $a^i_{\pos(i),l_i}$ according to (E5) and (E15). In the next step Spoiler picks up the pebble from $a^1_{\pos(1),l_1}$ and puts it on $b^2$. Duplicator has to answer with $b^2_{\pos(2)}$, mapping the edge $\{b^1,b^2\}$ to (E11). Because every vertex of the form $b^2_{0,l}$ is not connected to $b^1_{\pos(1)}$, this is the only choice for Duplicator. Thus, Spoiler can reach $\{(b^{1},b^1_{\pos(1)}),\ldots,(b^{k},b^k_{\pos(k)})\}$ and from there he reaches $\{(y^{1},y^1_{\pos(1)}),\ldots,(y^{k},y^k_{\pos(k)})\}$ with the same technique.

In order to derive the winning strategies for Duplicator in (ii) and (iii) we consider several total homomorphisms from Spoiler's to Duplicator's side. Consider the edges (E1), (E3) and (E7) connecting \blacknode vertices with  \whitenode vertices in one block of Duplicator's side. They can be used by Duplicator to pebble a \whitenode vertex when Spoiler moves upwards. This is the crucial ingredient for Duplicator's output strategies (ii).
To formalizes this let $H^{\text{out}}_{\pos}$ denote the set of total homomorphisms $h^\text{out}_{\pos,\sigma}\! =\! 
\{(x^i, x^i_0), 
(a^i, a^i_{\pos(i),\sigma(i)}), $ $
(b^i, b^i_{\pos(i)}), 
(y^i, y^i_{\pos(i)}) \mid i\!\in\![k]\}$,
where $\sigma \in S_k$ is some permutation on $[k]$. 
Furthermore, let $h^{\text{out}}_{(\pos,T)} = \{
(x^i, x^i_0), 
(a^i, a^i_{0}), $ $
(b^i, b^i_{0,i}),
(y^i, y^i_{(\pos,T)(i)})\mid i\in[k] \}
$.
Recall that $\cl(h)\defi \{g\mid g\subseteq h\}$. 
Since $h^{\text{out}}_{(\pos,T)}$ and all $h^\text{out}_{\pos,\sigma}\in H^{\text{out}}_{\pos}$ are total, 
\begin{align*}
  \mathcal H^\text{out}_{(\pos,T)} \defi \begin{cases}
                                         \cl(h^{\text{out}}_{(\pos,T)})\text{, } & T\neq \emptyset, \\
              \cl(\{h^{\text{out}}_{(\pos,\emptyset)}\}\cup H^{\text{out}}_{\pos}), & \text{else,}
                                        \end{cases}
\end{align*}
is a winning strategy for Duplicator satisfying (ii).

If a homomorphism maps all the $a^i$ vertices to $A^i_+$, then it has to map all $b^i$ vertices to $B^i_+$. This is due to the missing edges in (E5), (E15) and has also been used in Spoiler's strategy above. On the other hand, if at least one $a^i$ is mapped to $a^i_0$, then every $b^i$ can be mapped to $b^i_{0,l}$, where $l$ is chosen such that $a^j_{\pos(j),l}$ is not in the image of the homomorphism for every $j$. Duplicator benefits from this, because she can now map the $y^i$ vertices arbitrarily using the edges (E7). This behavior is used in the following restart strategies. Note that a homomorphism mapping some $a^i$ to $a^i_0$ also maps $x^i$ to $x^i_0$, hence restart strategies require invalid input positions.
For all nonempty $T\subseteq[k]$, let $\mathcal H^{\text{restart}}_{(\pos,T)} \defi \cl(H^{\text{restart}}_{(\pos,T)})$, where $H^{\text{restart}}_{(\pos,T)}$ is the set of total homomorphisms $h$ satisfying the constraints $h(x^i) = x^i_{(\pos,T)(i)}$ and $h(y^i) = y^i_0$.
This set clearly satisfies (iii). As an example let $g \in H^{\text{restart}}_{(\pos,\{t\})}$ be the following homomorphism:
\begin{align*}
g(x^i) &= x^i_{(\pos,\{t\})(i)},&g(b^i) &= b^i_{0,t},\\
g(a^t) &= a^t_0, &g(y^i) &= y^i_0,\\
g(a^i) &= a^i_{\pos(i),i}\text{, }i\neq t. 
\end{align*}
It remains to consider the critical input strategies (iv). They formalize the following behavior of Duplicator at the time when Spoiler wants to pebble a position $\pos$ through the switch as in (i). If Spoiler pebbles $a^i$ or $b^i$, Duplicator answers within $A^i_+$ or $B^i\setminus B^i_+$, respectively. This allows her to answer on the boundary according to the boundary function defined in (iv). However, she may run into trouble when Spoiler places $k$ pebbles on $a^i$ and $b^i$ vertices, because they extend to a $(k+1)$-clique on Spoiler's side, but not on Duplicator's side (on the blocks $A^i_+$ and $B^i\setminus B^i_+$). These positions form the critical positions where Duplicator switches to an output or restart strategy. If all $k$ pebbles are on $a_1,\ldots,a_k$, as in Spoiler's strategy (i), then Duplicator switches to the output strategy $\mathcal H^{\text{out}}_{\pos}$. In all other cases she switches to a restart strategy.
For all $l\in[k]$ and permutations $\sigma$ on $[k]$ we define partial homomorphisms:
\begin{align*}
h^{\text{in}}_{\pos,\sigma}(x^i) &= x^i_{\pos(i)} &
h^{\text{in}}_{\pos,\sigma,l}(x^i) &= x^i_{\pos(i)}    \\
h^{\text{in}}_{\pos,\sigma}(a^i) &=  a^i_{\pos(i),\sigma(i)} &
h^{\text{in}}_{\pos,\sigma,l}(a^i) &= a^i_{\pos(i),\sigma(i)}\text{, } i\neq\sigma^{-1}(l)\\
 h^{\text{in}}_{\pos,\sigma}(b^i) &= \text{undefined}&
h^{\text{in}}_{\pos,\sigma,l}(b^i) &= b^i_{0,l} \\
h^{\text{in}}_{\pos,\sigma}(y^i) &= y^i_{0}&
h^{\text{in}}_{\pos,\sigma,l}(y^i) &= y^i_{0}   
\end{align*}
It is easy to check, that $h^{\text{in}}_{\pos,\sigma}$ defines a homomorphism from $M^{k,n}_S\setminus B$ to $M^{k,n}_D$, and $h^{\text{in}}_{\pos,\sigma,l}$ defines a homomorphism from $M^{k,n}_S\setminus\{a^{\sigma^{-1}(l)}\}$ to $M^{k,n}_D$. 
For all $\sigma\in S_k$ 
$$
h^{\text{out-crit}}_{\pos,\sigma} \defi  \{(a^i,a^i_{\pos(i),\sigma(i)})\mid i\in [k]\}
$$
and for all $\sigma\in S_k$ and $j,t\in [k]$
$$
 h^{\text{restart-crit}}_{\pos,\sigma,j,t} \defi \{(a^i, a^i_{\pos(i),\sigma(i)})\mid i\in[k]\setminus \{t\}\}\cup\{(b^j,b^j_{0,\sigma(t)})\}.
$$
Now we can define the sets used in (iv):
\begin{align*}
 \mathcal H^\text{in}_{\pos} &= \cl(\{h^{\text{in}}_{\pos,\sigma}\mid \sigma\in S_k\}\cup\{h^{\text{in}}_{\pos,\sigma,l}\mid \sigma\in S_k, l\in [k]\}), \\
\mathcal C^\text{out-crit}_{\pos} &= \{h^{\text{out-crit}}_{\pos,\sigma}\mid \sigma\in S_k\}, \\
\mathcal C^\text{restart-crit}_{\pos,t} &= \{h^{\text{restart-crit}}_{\pos,\sigma,j,t}\mid \sigma\in S_k, j\in [k]\},\\
\crit(\mathcal H^\text{in}_{\pos}) &= \bigcup_{t\in[k]} \mathcal C^\text{restart-crit}_{\pos,t}\cup \mathcal C^\text{out-crit}_{\pos}.
\end{align*}
First note that $h^{\text{out-crit}}_{\pos,\sigma} \subset h^{\text{in}}_{\pos,\sigma}$ and $h^{\text{restart-crit}}_{\pos,\sigma,j,t}\subset h^{\text{in}}_{\pos,\sigma,\sigma(t)}$. Thus, $\crit(\mathcal H^\text{in}_{\pos})\subseteq \mathcal H^\text{in}_{\pos}$.
It easily follows from the definitions, that $h^{\text{out-crit}}_{\pos,\sigma}\subset h^{\text{out}}_{\pos,\sigma}$. Furthermore, every $h^{\text{restart-crit}}_{\pos,\sigma,l,t}$ can be extended to a homomorphism in $\mathcal H^\text{restart}_{(\pos,\{t\})}$ by defining the boundary as required and mapping $a^t$ to $a^t_0$ and $b^j$ to $b^j_{0,\sigma(t)}$ for all $j\in [k]$. This proves statement b) and c) from (iv). It remains to show that $\mathcal H^\text{in}_{\pos}$ is a critical strategy with critical positions $\crit(\mathcal H^\text{in}_{\pos})$.
\begin{claim}
 For all $g\in \mathcal H^\text{in}_{\pos}$ with $|\dom(g)|\leq k$, either $g\in \crit(\mathcal H^\text{in}_{\pos})$ or for all $z\in V(M^{k,n}_S)$ there exist an $h\in \mathcal H^\text{in}_{\pos}$, such that $g\subseteq h$ and $z\in\dom(h)$.
\end{claim}
\begin{innerproof} As $g$ is a partial homomorphism from $\mathcal H^\text{in}_{\pos}$, we can fix some $\sigma\in S_k$ and $l\in [k]$ such that $g\subset \{
 (x^i, x^i_{\pos(i)}),
 (a^i, a^i_{\pos(i),\sigma(i)}),
 (b^i, b^i_{0,l}),
 (y^i, y^i_{0})\mid i\in[k]\}$.

\noindent
\textbf{Case 1: $|\dom(g)\cap A|=k$.} 
In this case, $g=h^{\text{out-crit}}_{\pos,\sigma}$ and hence, $g\in\crit(\mathcal H^\text{in}_{\pos})$.

\noindent
\textbf{Case 2: $|\dom(g)\cap A|=k-1$.} 
If $\dom(g)\cap B \neq \emptyset$, then $g=h^{\text{restart-crit}}_{\pos,\sigma,j,\sigma^{-1}(l)}$ for some $j\in[k]$. Thus, we can assume that $\dom(g)\cap B = \emptyset$ and show for all $z$ that $g$ satisfies the extension property. If $z$ is the unique element in $A\setminus\dom(g)$, then $h^{\text{in}}_{\pos,\sigma}$ extends $g$. If $z\in X\cup B\cup Y$, then $h^{\text{in}}_{\pos,\sigma,l}$ extends $g$.

\noindent
\textbf{Case 3: $|\dom(g)\cap A|\leq k-2$.}
Let $j_1$ and $j_2$ be two distinct indices such that $a^{j_1}$, $a^{j_2}\notin \dom(g)$. Furthermore, we can assume that $\sigma(j_1) = l$. For $z\neq a^{j_1}$ the homomorphism $h^{\text{in}}_{\pos,\sigma,l}$ extends $g$. If $z=a^{j_1}$, then $h^{\text{in}}_{\pos,\sigma',l}$ extends $g$, where 
$\sigma' \defi \{(i,\sigma(i))\mid i\in [k]\setminus\{j_1,j_2\}\}\cup \{(j_1,\sigma(j_2)),(j_2,\sigma(j_1))\}.$
\end{innerproof}
\end{IEEEproof}

\subsection{The Initialization Gadget}

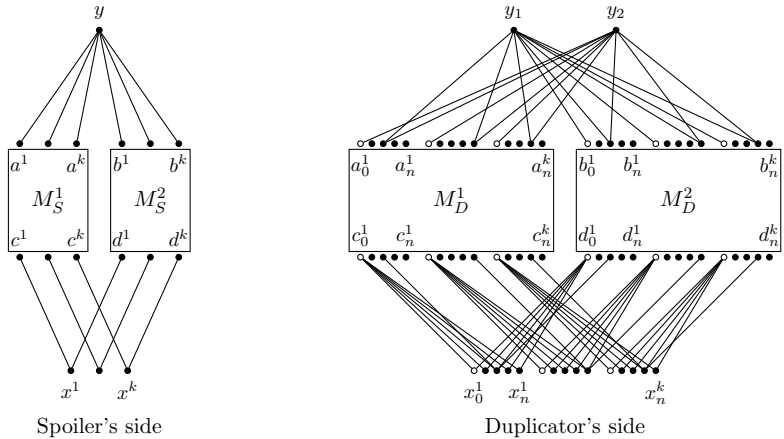
\begin{figure}[htp]
 \centering
  \input{Tikz/Tikz-INIT.tex}
 \caption{The initialization gadget INIT$^{\startpos}$ (for $k=3$, $n=4$, $\startpos(1)=2$, $\startpos(2)=4$, $\startpos(3)=3$.)} \label{fig:init}
\end{figure}

The initialization gadget {\upshape INIT$^{\startpos}$} is built out of two multiple input one-way switches $M^{1}$ and $M^{2}$, vertices $y$ in Spoiler's graph and $y_1$, $y_2$ in Duplicator's graph (that are colored $c_y$), and the boundary vertices $x^1,\ldots,x^k$ and $x^1_0,\ldots,x^k_n$ (where all $x^i$ and $x^i_l$ are colored $c_{x^i}$). The $x$- and $y$-vertices are connected to $M^1$ and $M^2$ as shown in Figure \ref{fig:init}. Lemma \ref{lem:init} (i)--(iii) provides the strategies on {\upshape INIT$^{\startpos}$}. The main property is that Spoiler can reach the start position at the boundary (i) and Duplicator has a corresponding counter strategy (ii) in this situation. 
Furthermore, if an arbitrary position occurs at the boundary during the game, Duplicator has a strategy to survive (iii).

\begin{lemma}\label{lem:init}
For every start position $\startpos\colon[k]\to[n]$ the following holds in the existential ($k+1$)-pebble game on {\upshape INIT$^{\startpos}$}:
\begin{enumerate}[leftmargin=2.2 em,label=(\roman*)]
 \item Spoiler can reach $\{(x^{i},x^{i}_{\startpos(i)})\mid i\in [k]\}$ from $\emptyset$.
 \item There is a winning strategy $\mathcal I^{\text{init}}$ for Duplicator with boundary function $\{(x^{i},x^{i}_{\startpos(i)})\mid i\in [k]\}$.
 \item For every $\pos\colon[k]\to [n]$ and every $T\subseteq [k]$ there is a critical strategy $\mathcal I^{\text{init}}_{(\pos,T)}$ with boundary function $\{(x^{i},x^{i}_{(\pos,T)(i)})\mid i\in [k]\}$ and $\crit(\mathcal I^{\text{init}}_{(\pos,T)})\subseteq \mathcal I^{\text{init}}$.
\end{enumerate}
\end{lemma}

\noindent Spoiler's strategy is quite simple. First he pebbles $y$. Duplicator has to answer with either $y_1$ or $y_2$. Then Spoiler can reach $\{(x^{i},x^{i}_{\startpos(i)})\mid i\in [k]\}$ by pebbling through either $M^1$ or $M^2$. To construct the strategies for Duplicator, we can combine the strategies of the switches $M^1$ and $M^2$ such that she plays an input strategy on one switch and a restart or output strategy on the other switch. Assume that Spoiler reaches a critical position on the switch where Duplicator plays the input strategy, say $M^1$. Duplicator can now flip the strategies such that she plays a restart or output strategy on $M^1$, depending on which kind of critical position Spoiler has reached, and an input strategy on $M^2$. 

\begin{IEEEproof}[Proof of Lemma \ref{lem:init}]
 We first develop the strategy for Spoiler (i). Spoiler first pebbles $y$. Duplicator has to response with either $y_{1}$ or $y_{2}$. Depending on Duplicator's choice, Spoiler can reach either $\{(a^{i},a^{i}_{\startpos(i)})\mid i\in[k]\}$ or $\{(b^{i},b^{i}_{\startpos(i)})\mid i\in[k]\}$. By Lemma \ref{lem:multi}.(i) Spoiler reaches $\{(c^{i},c^{i}_{\startpos(i)})\mid i\in[k]\}$ ($\{(d^{i},d^{i}_{\startpos(i)})\mid i\in[k]\}$) and from there he can reach the position $\{(x^{i},x^{i}_{\startpos(i)})\mid i\in[k]\}$.
For Duplicator's strategies we start with a discussion of possible moves on the boundary of the switches and the $x$- and $y$-vertices. At the top of the gadget Duplicator can map $y$ to $y_1$ and is then forced to answer with $\startpos$ at the input of $M^1$ and for some $R\subseteq[k]$ with $(\startpos,R)$ at the input of $M^2$. This strategy is called $\mathcal G^1_R$, the dual strategy where $y$ is mapped to $y_2$ is called $\mathcal G^2_R$.
\begin{align*}
\mathcal G^{1}_{R} &\defi \cl(\{(y,y_{1})\}\cup \{(a^{i},a^{i}_{\startpos(i)}),(b^{i},b^{i}_{(\startpos,R)(i)})\mid i\in [k]\}) \\
\mathcal G^{2}_{R} &\defi \cl(\{(y,y_{2})\}\cup \{(b^{i},b^{i}_{\startpos(i)}),(a^{i},a^{i}_{(\startpos,R)(i)})\mid i \in [k]\}) 
\end{align*}
At the bottom of the switch $\mathcal K_{(\pos,T)}$ denotes the strategy where $\boldsymbol 0$ is at the output of both switches and some arbitrary $(\pos,T)$ is at the $x$-block. In the strategy $\mathcal K^{\text{out-$i$}}$ the start position $\startpos$ occurs at the output of the switch $M^i$ and on the $x$-block, whereas $\boldsymbol 0$ is at the output of the other switch.   
\begin{align*}  
\mathcal K_{(\pos,T)} &\defi \cl(\{(c^{i},c^{i}_{0}), (d^{i},d^{i}_{0}), (x^{i},x^{i}_{(\pos,T)(i)})\mid i\in [k]\}) \\
\mathcal K^{\text{out-1}} &\defi \cl(\{(c^{i},c^{i}_{\startpos(i)}),(d^{i},d^{i}_{0}),(x^{i},x^{i}_{\startpos(i)})\mid i \in [k]\}) \\
\mathcal K^{\text{out-2}} &\defi \cl(\{(c^{i},c^{i}_{0}),(d^{i},d^{i}_{\startpos(i)}),(x^{i},x^{i}_{\startpos(i)})\mid i\in[k]\}) 
\end{align*}
Now we can combine these partial strategies with the strategies on the switches described in Lemma \ref{lem:multi}. In strategy $\mathcal I^{\text{in-$i$}}_{R,(\pos,T)}$ Duplicator plays an input strategy on switch $i$, a restart strategy on the other switch and an arbitrary position $(\pos,T)$ occurs at the $x$-block. These strategies were combined to the critical strategy $\mathcal I^{\text{init}}_{(\pos,T)}$ described in (ii).
\begin{align*}
\mathcal I^{\text{in-1}}_{R,(\pos,T)} &\defi  \mathcal G^{1}_{R} \circ \mathcal H^{\text{in}}_{\startpos}\langle M^{1}\rangle \circ \mathcal H^{\text{restart}}_{(\startpos,R)}\langle M^{2}\rangle \circ \mathcal K_{(\pos,T)} \\
\mathcal I^{\text{in-2}}_{R,(\pos,T)} &\defi  \mathcal G^{2}_{R} \circ  \mathcal H^{\text{restart}}_{(\startpos,R)}\langle M^{1}\rangle \circ \mathcal H^{\text{in}}_{\startpos}\langle M^{2}\rangle \circ \mathcal K_{(\pos,T)} \\
\mathcal I^{\text{init}}_{(\pos,T)} &\defi \bigcup_{t\in [k]} (\mathcal I^{\text{in-1}}_{\{t\},(\pos,T)}\cup \mathcal I^{\text{in-2}}_{\{t\},(\pos,T)})
\end{align*}
All critical positions of $\mathcal I^{\text{in-$i$}}_{R,(\pos,T)}$ are restart or output critical positions on the switch $M^i$. By Lemma \ref{lem:multi}.(iv).(b) every restart critical positions of $\mathcal I^{\text{in-1}}_{R,(\pos,T)}$ is contained in one of the strategies $\mathcal I^{\text{in-2}}_{\{t\},(\pos,T)}$ as non-critical position. Hence, the only critical positions $\crit(\mathcal I^{\text{init}}_{(\pos,T)})$ of the combined strategy are output critical positions on the switches. These output critical positions will be contained in the strategies $\mathcal I^{\text{init-$i$}}$ where Duplicator plays an output strategy on switch $i$. Together with $\mathcal I^{\text{init}}_{\startpos}$ they form the winning strategy $\mathcal I^{\text{init}}$ from (ii).
\begin{align*}  
\mathcal I^{\text{init-1}} &\defi \mathcal G^{2}_{[k]} \circ \mathcal H^{\text{out}}_{\startpos}\langle M^{1}\rangle \circ \mathcal H^{\text{in}}_{\startpos}\langle M^{2}\rangle \circ \mathcal K^{\text{out-1}} \\
\mathcal I^{\text{init-2}} &\defi \mathcal G^{1}_{[k]} \circ \mathcal H^{\text{in}}_{\startpos}\langle M^{1}\rangle \circ \mathcal H^{\text{out}}_{\startpos}\langle M^{2}\rangle \circ \mathcal K^{\text{out-2}} \\
\mathcal I^{\text{init}} &\defi \mathcal I^{\text{init-1}} \cup \mathcal I^{\text{init-2}} \cup  \mathcal I^{\text{init}}_{\startpos}
 \end{align*}
 $\mathcal I^{\text{init}}$ is a union of critical strategies with $\startpos$ at the boundary.
To prove that $\mathcal I^{\text{init}}$ is indeed a winning strategy on the gadget, we apply Lemma \ref{lem:combunion} and show that every critical position of one strategy is contained as non-critical position in another strategy. 
Critical positions are inside the input strategy $\mathcal H^{\text{in}}_{\startpos}$ on one of the switches. 
By Lemma \ref{lem:multi}.(iv) they are either contained in an output or restart strategy on the corresponding switch. Hence, all restart critical positions on $M^1$ and $M^2$ are contained in $\mathcal I^{\text{init}}_{\startpos}$ and all output critical positions on $M^1$ ($M^2$) are contained in $\mathcal I^{\text{init-1}}$ ($\mathcal I^{\text{init-2}}$). 
Recall that $\mathcal{\hat S} \defi \mathcal S\setminus \crit(\mathcal S)$, by Lemma \ref{lem:multi}.(iv) we get:
\begin{align*}
\crit(\mathcal I^{\text{in-2}}_{R,(\pos,T)}) = \crit(\mathcal I^{\text{init-1}}) = \crit(\mathcal H^{\text{in}}_{\startpos}\langle M^{2}\rangle)
&\subseteq \mathcal H^\text{out}_{\startpos}\langle M^{2}\rangle \cup \bigcup_{t\in[k]} \mathcal H^\text{restart}_{(\startpos,\{t\})}\langle M^{2}\rangle \\
&\subseteq \mathcal {\hat I}^{\text{init-2}} \cup \bigcup_{t\in[k]} \mathcal{\hat I}^{\text{in-1}}_{\{t\},\startpos}, \\
\crit(\mathcal I^{\text{in-1}}_{R,(\pos,T)}) = \crit(\mathcal I^{\text{init-2}}) = \crit(\mathcal H^{\text{in}}_{\startpos}\langle M^{1}\rangle) 
&\subseteq \mathcal H^\text{out}_{\startpos}\langle M^{1}\rangle \cup \bigcup_{t\in[k]} \mathcal H^\text{restart}_{(\startpos,\{t\})}\langle M^{1}\rangle \\
&\subseteq \mathcal {\hat I}^{\text{init-1}} \cup \bigcup_{t\in[k]} \mathcal{\hat I}^{\text{in-2}}_{\{t\},\startpos}. 
\end{align*}
Hence, $\crit(\mathcal I^{\text{init}}_{(\pos,T)})\subseteq \mathcal I^{\text{init}}$ and $\mathcal I^{\text{init}}$ is a winning strategy by Lemma \ref{lem:combunion}.
\end{IEEEproof}

\subsection{The Choice Gadget}

The boundary of the choice gadget consists of input vertices $x^1,\ldots,x^k$ in Spoiler's graph and $x^1_0,\ldots,x^k_n$ in Duplicator's graph. These vertices are identified with $\ynode$-vertices in the final graph. The output vertices are of the form $(y_q)^1,\ldots,(y_q)^k$ and $(y_q)^1_0,\ldots,(y_q)^k_n$ for all $q\in[m]$ and are connected to the rule gadgets $RD(r_q)$. This gadget enables Spoiler to reach positions $((y_q)^i,(y_q)^i_{\pos(i)})$ from $(x^i,x^i_{\pos(i)})$ but Duplicator can choose the desired $q\in[m]$. This choice will later coincide with the rule $r_q\in R$ Player~2 chooses in the KAI-game when position $\pos$ is pebbled.
The \textit{choice gadget} $C^{m}$ is defined as follows:
\begin{align*}
&V(C^{m}_S) = \{x^{i},a^{i}\mid i\in[k]\}  \cup \{(y_q)^{i}\mid i\in[k], q\in [m]\},\\
&E(C^{m}_S) = \big\{\{x^{i},a^{i}\}\mid i\in [k]\big\} \cup \\ 
            &\qquad\big\{\{a^{i},(y_q)^{i}\}\mid i\in [k],q\in [m]\big\} \\
           &\qquad \cup \big\{\{a^{i},a^{j}\}\mid i,j\in [k]\text{; } i\neq j\big\},  \\
&V(C^{m}_D) =  \{x^{i}_l\mid i\in[k],0\leq l\leq n\} \cup  \{a^{i}_0\mid i\in[k]\} \\
           &\qquad  \{a^{i}_{l,q}\mid i\in[k], l\in [n], q\in [m]\} \\
           &\qquad\cup  \{(y_q)^{i}_l\mid i\in[k], q\in [m], 0\leq l\leq n\},\\
&E(C^{m}_D) = \big\{\{x^{i}_0,a^{i}_0\}\mid i\in [k]\big\} \\
           &\qquad \cup \big\{\{x^{i}_l,a^{i}_{l,q}\}\mid i\in [k],l\in[n],q\in[m]\big\} \\
           &\qquad \cup \big\{\{a^{i}_0,(y_q)^{i}_0\}\mid i\in [k],q\in[m]\big\} \\
           &\qquad \cup \big\{\{a^{i}_{l,q},(y_q)^{i}_l\}\mid i\in [k];l\in[n];q\in[m]\big\} \\
           &\qquad \cup \big\{\{a^{i}_{l,q},(y_{q'})^{i}_0\}\mid i\in [k];l\in[n];q,q'\in[m];q\neq q'\big\} \\
           &\qquad \cup \big\{\{a^{i}_{l,q},a^{j}_{l',q}\}\mid i,j\in [k]; i\neq j;l,l'\in[n];q\in[m]\big\}.   
\end{align*}
Furthermore, for all $i\in [k]$, all vertices $a^i$ and $a^i_{l,q}$ are colored with the unique color $c_{a^i}$, and for all $i\in[k]$, $q\in[m]$ the vertices $(y_q)^i$ and $(y_q)^i_l$ are colored with the unique color $c_{(y_q)^i}$. 
One partition of $C^m$ is shown in Figure \ref{fig:choiceD}.
\begin{figure}[htp]
 \centering
\input{Tikz/Tikz-Choice-LICS.tex}
 \caption{The $i$-th partition of the gadget $C^{m}$.}\label{fig:choiceD}
\end{figure}
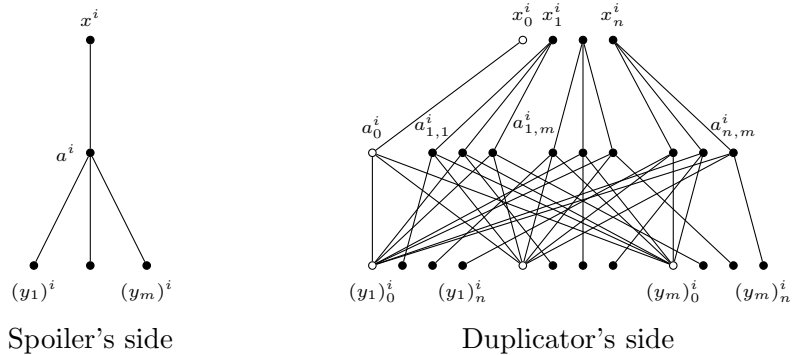
Recall that Spoiler can reach a position if he has a strategy such that either this position occurs after some finite number of rounds or he wins the game. We extend this notion to sets of positions by saying that Spoiler can reach one of the positions $p_1,\ldots,p_m$ from $p_0$ if starting from position $p_0$ either Spoiler wins the game or one of the positions $p_1,\ldots,p_m$ occurs after a finite number of rounds.
\begin{lemma} \label{lem:choice}
In the existential ($k+1$)-pebble game on $C^{m}$, 
\begin{enumerate}[leftmargin=2.2 em,label=(\roman*)]
 \item for  every $\pos\colon[k]\to[n]$ Spoiler can reach one of the positions $\{((y_l)^i,(y_l)^i_{\pos(i)})\mid i\in[k]\}$ for $l\in [m]$ from $\{(x^i,x^i_{\pos(i)})\mid i\in[k]\}$, and
\item for every $l\in[m]$, $\pos\colon[k]\to[n]$, $T\subseteq [k]$, Duplicator has a winning strategy $\mathcal C^l_{(\pos,T)}$ with boundary function $\{(x^i,x^i_{(\pos,T)(i)})\mid i\in[k]\}\cup \{((y_q)^i,(y_q)^i_{0})\mid i\in[k],q\in[m]\setminus\{l\}\}\cup \{((y_l)^i,(y_l)^i_{(\pos,T)(i)})\mid i\in[k]\}$.
\end{enumerate}
\end{lemma}
\begin{IEEEproof}
We first present Spoiler's strategy (i). Starting from position $\{(x^i,x^i_{\pos(i)})\mid i\in[k]\}$ Spoiler puts the ($k+1$)st pebble on $a^{1}$ and Duplicator has to response with $a^{1}_{\pos(1),l}$ for one $l\in[m]$ she can choose. Then Spoiler picks up the pebble from $x^{1}$ and puts it on $a^{2}$. Duplicator has to response with $a^{2}_{\pos(2),l}$, since this vertex is the only one adjacent to $x^{2}_{\pos(2)}$ and $a^{1}_{\pos(1),l}$. Next, Spoiler picks up the pebble from $x^{2}$ and puts it on $a^{3}$ and so on. Thus, Spoiler reaches $\{(a^{i},a^{i}_{\pos(i),l})\mid i\in[k]\}$. In the next step he places the ($k+1$)st pebble on $(y_l)^{1}$, and Duplicator has to answer with $(y_l)^{1}_{\pos(1)}$, since this vertex is the only one colored in the same color as $(y_l)^{1}$ and adjacent to $a^{1}_{\pos(1),l}$. Now Spoiler picks up the pebble from $a^{1}$ and puts it on $(y_l)^{2}$, Duplicator has to answer with $(y_l)^{2}_{\pos(2)}$. Following that strategy Spoiler can reach  $\{((y_l)^i,(y_l)^i_{\pos(i)})\mid i\in[k]\}$.
Duplicator's strategy (ii) is simply defined as $\mathcal C^l_{(\pos,T)} = \cl(h^l_{(\pos,T)})$, where $h^l_{(\pos,T)}$ is the following total homomorphism on $C^{m}$:
\begin{align*}
h^l_{(\pos,T)}(z) &\defi \begin{cases}
x^i_{(\pos,T)(i)}, &\text{ if } z= x^{i},\\
a^{i}_{\pos(i),l}, &\text{ if } z= a^{i},i\notin T,\\
a^{i}_{0}, &\text{ if } z= a^{i},i\in T,\\
(y_l)^i_{(\pos,T)(i)}, &\text{ if } z= (y_l)^i,\\
(y_q)^i_{0}, &\text{ if } z= (y_q)^i\text{, }q\in[m]\setminus\{l\}.
\end{cases}
\end{align*}
Hence, if some (not necessarily valid) position $(\pos,T)$ is on the input Duplicator can force Spoiler to bring this position to the output block $(y_l)$ using the strategy $\mathcal C^l_{(\pos,T)}$.
\end{IEEEproof}
The strategies $\mathcal C^1_{\boldsymbol{0}} = \cdots = \mathcal C^m_{\boldsymbol{0}}$ have the $\boldsymbol{0}$ position at every vertex block and will be denoted by $\mathcal C_{\boldsymbol{0}}$.

\subsection{Proof of the Main Lemma}\label{sec:mainlemma}

\begin{IEEEproof}[Proof of Lemma \ref{lem:reduction}] It remains to construct winning strategies for Spoiler and Duplicator on the colored graphs $G_S$ and $G_D$. 
First, we develop the winning strategy for Spoiler. Hence assume that Player~1 has a winning strategy in the $k$-pebble KAI-game. By Lemma \ref{lem:init}, Spoiler can reach position $\{(\xnode^{i},\xnode^{i}_{\startpos(i)})\mid i\in [k]\}$. Let $r$ be the first rule Player 1 chooses in the KAI-game. By Lemma \ref{lem:srule} Spoiler can reach $\{(y^i,y^i_{r(\startpos)(i)})\mid i\in[k]\}$, where the $y$ vertices are the output-boundary of the rule gadget $RS(r)$. By Lemma \ref{lem:multi} he can pebble through the switch and reach $\{(\ynode^i,\ynode^i_{r(\startpos)(i)})\mid i\in[k]\}$. Let $\pos\defi r(\startpos)$. If $\pos$ maps some pebble to $\goal$, then Spoiler wins the game since $\ynode^i$ is colored $c_{y^{i}}$ whereas $\ynode^i_{\goal}$ is not. Otherwise (by Lemma \ref{lem:choice}) Spoiler can reach position $\{((y_q)^i,(y_q)^i_{\pos(i)})\mid i\in[k]\}$ for one $q\in[m]$ of Duplicator's choice at the output of the choice gadget and the input of RD($r_q$). If $r_q \notin \appl(\pos)$, then Spoiler wins immediately by Lemma \ref{lem:drule} (especially Spoiler wins if $\appl(\pos)=\emptyset$ and Player~2 cannot  move). If $r_q \in \appl(\pos)$, then Spoiler can reach $\{(\xnode^{i},\xnode^{i}_{r_q(\pos)(i)})\mid i\in [k]\}$. Spoiler chooses the next rule according to Player~1's winning strategy and so on. Since Player 1 eventually puts a pebble $c$ on node $\goal$, Spoiler can reach position $(\ynode^{c},\ynode^{c}_\goal)$ and thus he wins the game. 

Assume that $\mathcal K=(\mathcal K_1, \mathcal K_2,\kappa)$ is a winning strategy for Player 2 in the $k$-pebble KAI-game. Recall that Player 2 can play in such a way that every position occurring in the KAI-game is either contained in $\mathcal K_1$ or $\mathcal K_2$ where $\mathcal K_i$ is the set of position when it is Player $i$'s turn. We define a global critical strategy $\mathcal S_{\pos}$ ($\mathcal D_{\pos}$) for every position $\pos$ in $\mathcal K_1$ ($\mathcal K_2$). Duplicator can now simulate Player 2's winning strategy in the KAI-game by playing according to the critical strategy $\mathcal S_{\pos}$ ($\mathcal D_{\pos}$) if the position $\pos$ is the current position in the KAI-game and it is Player 1's (Player 2's) turn. If Spoiler pebbles output critical positions in these strategies, Duplicator switches the strategies in the same way as the positions in the KAI-game change. If Spoiler plays incorrectly in the sense that he pebbles a restart critical position at the switches, then Duplicator moves to a corresponding restart strategy.

Now we construct these critical strategies for the whole graph out of smaller critical strategies $\mathcal F$ defined on gadgets $Q$ (denoted $\mathcal F\langle Q\rangle$) using the $\circ$-operator and Lemma \ref{lem:combcirc}. The global strategy $\mathcal S^{\text{init}}$ means ``the KAI-game has just started, position $\startpos$ is on the board and it is Player 1's turn.'' The strategy $\mathcal S_{\pos}$ ($\mathcal D_{\pos}$) denotes ``position $\pos$ is on the board and it is Player 1's (Player 2's) turn.''
\begin{align*}
    \mathcal S^{\text{init}} &= 
    \mathcal I^{\text{init}} \circ  \mathcal C_{\boldsymbol{0}} \circ 
    \bigcirc_{l\in \appl(\startpos)} \left(\mathcal R_{\startpos}\langle RS(r_l)\rangle\circ\mathcal H^{\text{in}}_{r_l(\startpos)}\langle MS(r_l)\rangle \right) \circ 
  \\&\quad
    \bigcirc_{l\in [m]\setminus\appl(\startpos)} \left(\mathcal R_{\startpos}\langle RS(r_l)\rangle\circ\mathcal H^{\text{restart}}_{(r_l(\startpos),T_{r_l})}\langle MS(r_l)\rangle \right)  \circ 
  \\&\quad
    \bigcirc_{l\in[m]} \left(\mathcal R_{\boldsymbol{0}}\langle RD(r_l)\rangle\circ\mathcal H^{\text{out}}_{\startpos}\langle MD(r_l)\rangle\right).
\end{align*}
We define the global critical strategies $\mathcal S_{\pos}$ and $\mathcal S^{\text{restart}}_{(\pos,T)}$ for all $\pos\in \mathcal K_1$ and $T\neq\emptyset$. In the strategy $\mathcal S_{\pos}$ the position $\pos$ is at the $\xnode$-vertices encoding that $\pos$ is the current position in the KAI-game and it is Player 1's turn. The strategies $\mathcal S^{\text{restart}}_{(\pos,T)}$ contain the restart critical positions of $\mathcal S_{\pos}$.
\begin{align*}
 \mathcal S_{\pos} &= \mathcal I^{\text{init}}_{\pos} \circ  \mathcal C_{\boldsymbol{0}} \circ 
    \bigcirc_{l\in \appl(\pos)} \left(\mathcal R_{\pos}\langle RS(r_l)\rangle\circ\mathcal H^{\text{in}}_{r_l(\pos)}\langle MS(r_l)\rangle \right) \circ \\
    &\bigcirc_{l\in [m]\setminus\appl(\pos)} \left(\mathcal R_{\pos}\langle RS(r_l)\rangle\circ\mathcal H^{\text{restart}}_{(r_l(\pos),T_{r_l})}\langle MS(r_l)\rangle \right)  \circ\\
    &\bigcirc_{l\in[m]} \left(\mathcal R_{\boldsymbol{0}}\langle RD(r_l)\rangle\circ\mathcal H^{\text{out}}_{\pos}\langle MD(r_l)\rangle\right),\\
\mathcal S^{\text{restart}}_{(\pos,T)} &= \mathcal I^{\text{init}}_{(\pos,T)} \circ  \mathcal C_{\boldsymbol{0}} \circ \\
&\bigcirc_{l\in[m]} \big(\mathcal R_{(\pos,T)}\langle RS(r_l)\rangle \circ \mathcal H^{\text{restart}}_{(r_l(\pos),T\cup T_{r_l})}\langle MS(r_l)\rangle \circ\\
&\quad\mathcal R_0\langle RD(r_l)\rangle \circ \mathcal H^{\text{out}}_{(\pos,T)}\langle MD(r_l)\rangle\big).
\end{align*}
Furthermore, for all $\pos\in\mathcal K_2$ and $T\neq\emptyset$ let $\mathcal D_{\pos}$ and $\mathcal D^{\text{restart}}_{(\pos,T)}$ be the following global critical strategies. Similar as in the strategies above, $\mathcal D_{\pos}$ puts the position $\pos$ at the $\ynode$-vertices encoding that $\pos$ is the current position in the KAI-game and it is Player 2's turn. Again, the strategies $\mathcal D^{\text{restart}}_{(\pos,T)}$ contain the restart critical positions of $\mathcal D_{\pos}$.
\begin{align*}
\mathcal D_{\pos} &=  \mathcal I^{\text{init}}_{{\boldsymbol{0}}} \circ \mathcal C^{\choice(\pos)}_{\pos}\circ 
    \bigcirc_{l\in[m]} \big(\mathcal R_{\boldsymbol{0}}\langle RS(r_l)\rangle \circ \mathcal H^{\text{out}}_{\pos}\langle MS(r_l)\rangle\big)\circ\\
    &\bigcirc_{l\in[m]\setminus\{\choice(\pos)\}} \big(\mathcal R_{\boldsymbol{0}}\langle RD(r_l)\rangle \circ \mathcal H^{\text{restart}}_{{\boldsymbol{0}}}\langle MD(r_l)\rangle\big)\circ\\
    &\mathcal R_{\pos}\langle RD(r_{\choice(\pos)})\rangle \circ \mathcal H^{\text{in}}_{r_{\choice(\pos)}(\pos)}\langle MD(r_{\choice(\pos)})\rangle, \\
 \mathcal D^{\text{restart}}_{(\pos,T)} &=  \mathcal I^{\text{init}}_{{\boldsymbol{0}}} \circ \mathcal C^{\choice(\pos)}_{(\pos,T)}\circ 
    \bigcirc_{l\in[m]} \big(\mathcal R_{\boldsymbol{0}}\langle RS(r_l)\rangle \circ \mathcal H^{\text{out}}_{(\pos,T)}\langle MS(r_l)\rangle\big)\circ\\
    &\bigcirc_{l\in[m]\setminus\{\choice(\pos)\}} \big(\mathcal R_{\boldsymbol{0}}\langle RD(r_l)\rangle \circ \mathcal H^{\text{restart}}_{{\boldsymbol{0}}}\langle MD(r_l)\rangle\big)\circ\\
    &\mathcal R_{(\pos,T)}\langle RD(r_{\choice(\pos)})\rangle \circ 
    \mathcal H^{\text{restart}}_{(r_{\choice(\pos)}(\pos),T)}\langle MD(r_{\choice(\pos)})\rangle. 
\end{align*}
Before we formally state Duplicator's winning strategy, we briefly describe these critical strategies. First, the only critical positions of $\mathcal S^{\text{restart}}_{(\pos,T)}$ and $\mathcal D^{\text{restart}}_{(\pos,T)}$ are inside the initialization gadget and contained in $\mathcal S^{\text{init}}$. Thus, this is a good situation for Duplicator, since Spoiler has to restart the game by playing on the initialization gadget. 
At the beginning of the game Duplicator plays according to the strategy $\mathcal S^{\text{init}}$, where position $\startpos$ is on the $\xnode$-vertices. The only critical positions of that strategy are inside the $MS(r)$-gadgets for rules $r$ applicable to $\startpos$. If Spoiler pebbles some restart-critical positions there, then Duplicator can switch to $\mathcal S^{\text{restart}}_{(\startpos,T)}$.
If Spoiler pebbles an output-critical position on $MS(r)$, then Duplicator can switch to strategy $\mathcal D_{r(\startpos)}$ where position $\pos = r(\startpos)$ is on the $\ynode$-vertices. The only critical positions now are inside the switch $MD(r_{\choice(\pos)})$ and inside the initialization gadget. If Spoiler pebbles a restart-critical position on $MD(r_{\choice(\pos)})$, then Duplicator sticks to $ \mathcal D^{\text{restart}}_{(\pos,T)}$. If Spoiler pebbles an output-critical position, then Duplicator chooses strategy $\mathcal S_{\pos'}$, where $\pos'=r_{\choice(\pos)}(\pos)$. The critical positions from $\mathcal S_{\pos'}$ are within the initialization gadget and the switches $MS(r)$ for rules $r$ applicable to $\pos'$. Thus, combining all these critical strategies allows Duplicator to play forever.
Now we define the winning strategy for Duplicator: 
\begin{align*}
 \mathcal H = \mathcal S^{\text{init}} &\cup \bigcup_{\pos\in\mathcal K_1, T\subseteq [k],T\neq\emptyset} \left(\mathcal S_{\pos} \cup \mathcal S^{\text{restart}}_{(\pos,T)} \right)\\
&\cup \bigcup_{\pos\in \mathcal K_2, T\subseteq [k],T\neq\emptyset} \left(\mathcal D_{\pos} \cup \mathcal D^{\text{restart}}_{(\pos,T)} \right).
\end{align*}
Since $\mathcal H$ is a union of critical strategies, it suffices by Lemma \ref{lem:combunion} to show that for each critical strategy $\mathcal G$ and each partial homomorphism $h\in \crit(\mathcal G)$ there is a critical strategy $\mathcal F$ such that $h\in \mathcal{\hat{F}}$. From the definition of the global critical strategies and the properties of the partial critical strategies they contain, it follows that
\begin{align*}
 \crit(\mathcal S^{\text{restart}}_{(\pos,T)}) &\subseteq \mathcal {\hat{S}}^\text{init},\\ 
\crit(\mathcal D^{\text{restart}}_{(\pos,T)}) &\subseteq \mathcal {\hat{S}}^\text{init},\\ 
 \crit(\mathcal S_{\pos}) &\subseteq \mathcal {\hat{S}}^{\text{init}} \cup \bigcup_{T\neq\emptyset} \mathcal {\hat{S}}^{\text{restart}}_{(\pos,T)} \cup \bigcup_{l\in\appl(\pos)} \mathcal {\hat{D}}_{r_l(\pos)}, \\
 \crit(\mathcal D_{\pos}) &\subseteq \mathcal {\hat{S}}^\text{init}\cup  \mathcal {\hat{S}}_{r_{\choice(\pos)}(\pos)} \cup \bigcup_{T\neq\emptyset} \mathcal {\hat{D}}^{\text{restart}}_{(\pos,T)} , \\
 \crit(\mathcal S^{\text{init}}) &\subseteq \bigcup_{l\in\appl(\startpos)} \mathcal {\hat{D}}_{r_l(\startpos)}.
\end{align*}
From the definition of $\mathcal H$ and the properties of the KAI-game winning strategy $\mathcal K$, it follows, that if a global critical strategy mentioned in the left hand side of the above inclusions is a strategy in $\mathcal H$, then so are all strategies on the right hand side. This concludes the proof of Lemma \ref{lem:reduction} for colored simple graphs $G_S$ and $G_D$. 
\end{IEEEproof}

\subsection{Getting Rid of the Colors}

As in \cite{Kolaitis.2003} our construction involves $|V(G_S)|$ unary predicates. To settle the complexity of deciding whether Spoiler has a winning strategy in the existential $k$-pebble game on $\sigma$-structures for fixed finite signatures $\sigma$, we use the following construction to switch from colored simple graphs to directed graphs. Let $P_1,\ldots,P_w$ be the colors used in the graphs $G_S$ and $G_D$, and $P$ an additional color. We introduce $w+1$ vertices $d_0,\ldots,d_{w}$ that are colored $P$ in both graphs $G_S$ and $G_D$. For every $1\leq i < w$ and each vertex $x$ colored $P_i$, there are directed edges $(d_i,x)$ and $(x,d_{i+1})$. Furthermore, there are directed edges from vertices colored $P_w$ to $d_w$ and one from $d_1$ to $d_0$. Now we can delete the colors $P_1,\ldots,P_w$ without giving Duplicator more freedom. 
That is, we replace the requirement ``$x\in P_i$'' by the statement ``there exists an alternating directed path (where every second element is colored $P$) of length $2i-1$ to a vertex colored $P$ and having an out-neighbor colored $P$''. This can easily be checked by Spoiler using two pebbles.
To get rid of the remaining color $P$ we add a loop $E(x,x)$ for every vertex $x\in P$.

\begin{figure}[htp]
 \centering

\begin{tikzpicture}
  [scale=1, transform shape, knoten/.style={circle,draw=black,fill=black,
  inner  sep=0pt,minimum  size=1mm},kknoten/.style={circle,draw=black,fill=black,
  inner  sep=0pt,minimum  size=0.8mm}, leerknoten/.style={circle,draw=black,fill=white,
  inner  sep=0pt,minimum  size=1mm},leerkknoten/.style={circle,draw=black,fill=white,
  inner  sep=0pt,minimum  size=0.8mm}]
  
  \node[knoten] (d0) at (0,4.5)[label=left:{\tiny{$d_0$}}] {};
  \node[knoten] (d1) at (0,4)[label=left:{\tiny{$d_1$}}] {};
  \node[knoten] (d2) at (0,3) {};
  \node[knoten] (d3) at (0,2) {};
  \node[knoten] (d4) at (0,1)[label=left:{\tiny{$d_{w}$}}]{};
  
  \foreach \i/\max in {1/0,2/0,3/0,4/0} \foreach \j in {0}{
  \node[knoten] (p\i\j) at (0.5+\j*0.2,4.5-\i) {};
  }
  
  \foreach \i/\max/\iplusone in {1/2/2,2/4/3,3/6/4,4/3/5} \foreach \j in {0}{
  \draw[->] (p\i\j)--(d\i);
  } 
  
  \foreach \i/\max/\iplusone in {1/2/2,2/4/3,3/6/4} \foreach \j in {0}{
  \draw[<-] (p\i\j)--(d\iplusone);
  }

 \draw[<-] (d0)--(d1); 
 
 \draw[rounded corners] (-0.15,4.6) rectangle (0.15,0.9);
 \node at (-0.35,2.5) {\small$P$};
 
  \draw[rounded corners] (0.35,3.35) rectangle (0.65,3.65);
 \node at (0.9,3.5) {\small$P_1$};
 
   \draw[rounded corners] (0.35,2.35) rectangle (0.65,2.65);
 \node at (0.9,2.5) {\small$P_2$};
 
   \draw[rounded corners] (0.35,1.35) rectangle (0.65,1.65);

   \draw[rounded corners] (0.35,0.35) rectangle (0.65,0.65);
 \node at (0.9,0.5) {\small$P_w$};
 
 \node at (0.5,-0.5) {$G_S$};
 
\begin{scope}[xshift=4cm]
  \node[knoten] (d0) at (0,4.5)[label=left:{\tiny{$d_0$}}] {};
  \node[knoten] (d1) at (0,4)[label=left:{\tiny{$d_1$}}] {};
  \node[knoten] (d2) at (0,3) {};
  \node[knoten] (d3) at (0,2) {};
  \node[knoten] (d4) at (0,1)[label=left:{\tiny{$d_{w}$}}] {};
  
  \foreach \i/\max in {1/3,2/6,3/6,4/4} \foreach \j in {0,1,2,...,\max}{
  \node[knoten] (p\i\j) at (0.5+\j*0.2,4.5-\i) {};
  }
  
    \foreach \i/\max/\iplusone in {1/3/2,2/6/3,3/6/4,4/4/5} \foreach \j in {0,1,2,...,\max}{
  \draw[->] (p\i\j)--(d\i);
  } 
  
  \foreach \i/\max/\iplusone in {1/3/2,2/6/3,3/6/4} \foreach \j in {0,1,2,...,\max}{
  \draw[<-] (p\i\j)--(d\iplusone);
  } 

 \draw[->] (d1)--(d0); 
 
 \draw[rounded corners] (-0.15,4.6) rectangle (0.15,0.9);
 \node at (-0.35,2.5) {\small$P$};
 
  \draw[rounded corners] (0.35,3.35) rectangle (1.25,3.65);
 \node at (1.45,3.5) {\small$P_1$};
 
   \draw[rounded corners] (0.35,2.35) rectangle (1.85,2.65);
 \node at (2.05,2.5) {\small$P_2$};
 
   \draw[rounded corners] (0.35,1.35) rectangle (1.85,1.65);

   \draw[rounded corners] (0.35,0.35) rectangle (1.45,0.65);
 \node at (1.65,0.5) {\small$P_w$};
 
  \node at (1,-0.5) {$G_D$};
\end{scope}
  
\end{tikzpicture}

 \caption{Colors $P_1,\ldots,P_w$ can now be deleted.}\label{fig:color}
\end{figure}
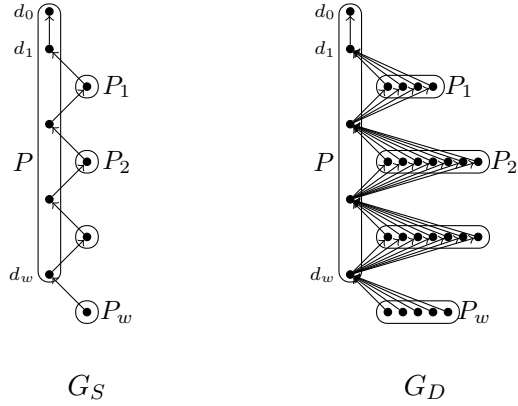

\section{Conclusion}

We proved an $\Omega(n^{\frac{(k-2)}{12}-\epsilon})$ lower bound for deciding which player can win the existential $k$-pebble game on two given finite relational structures and got an $\Omega(n^{\frac{(k-2)}{12}-\epsilon})$ lower bound for the $k$-consistency test as a consequence. 
Furthermore, the lower bound also applies to the $k$-pebble game that characterizes the expressive power of the existential $k$-variable fragment (where negation is allowed in front of atomic formulas). 

The parameterized complexity of the whole $k$-variable first-order logic $\mathsf L^k$ and the counting logic $\mathsf C^k$ is wide open. It is not even known if it is W[1]-hard to decide if two given finite relational structures can be distinguished by an $\mathsf L^k$  ($\mathsf C^k$) sentence.
Regarding the classical complexity, $\mathsf L^k$-equivalence as well as $\mathsf C^k$-equivalence is complete for polynomial time \cite{Grohe.1996}, but it is an open problem whether the problems are
complete for EXPTIME when $k$ is part of the input.

\bibliography{Literatur}
\bibliographystyle{plain}\enlargethispage{12 pt}
\section*{Acknowledgement}

I want to thank Martin Grohe and Martin Otto for a fruitful discussion on an earlier version of this paper. Furthermore, I thank the anonymous referees for the useful comments that helped to improve this paper.\vspace{-24 pt}

\end{document}

%% file: Tikz/Tikz-Glueing-LICS.tex
\begin{tikzpicture}
	[scale=1, transform shape, knoten/.style={circle,draw=black,fill=black,
	inner  sep=0pt,minimum  size=1mm},kknoten/.style={circle,draw=black,fill=black,
	inner  sep=0pt,minimum  size=0.8mm}, leerknoten/.style={circle,draw=black,fill=white,
	inner  sep=0pt,minimum  size=1mm},leerkknoten/.style={circle,draw=black,fill=white,
	inner  sep=0pt,minimum  size=0.8mm}]



	\draw (12,3.5) rectangle (13.5,6.1);
	\node[rotate=90] at (12.75,4.8) {\upshape INIT$^{\startpos}$};
	\node[rotate=90] at (11.7+0.5,3.7) {\tiny $x^1_0$};
	\node[rotate=90] at (11.7+0.5,5.9) {\tiny $x^k_n$};
	
	\draw (0.5,7.1) rectangle (3.1,8.4);
	\draw (3.5,7.1) rectangle (6.1,8.4);
	\draw (7.9,7.1) rectangle (10.5,8.4);
	\draw (10.9,7.1) rectangle (13.5,8.4);
	\node at (1.8,7.75) {$\uparrow RS(r_1)$};
	\node at (4.8,7.75) {$\uparrow RS(r_m)$};
	\node at (9.2,7.75) {$\downarrow MD(r_1)$};
	\node at (12.2,7.75) {$\downarrow MD(r_m)$};
	
	\node at (0.7,7.3) {\tiny $x^1_0$};
	\node at (2.9,7.3) {\tiny $x^k_n$};
	\node at (0.7,8.2) {\tiny $y^1_0$};
	\node at (2.9,8.2) {\tiny $y^k_n$};

	\node at (0.7+3,7.3) {\tiny $x^1_0$};
	\node at (2.9+3,7.3) {\tiny $x^k_n$};
	\node at (0.7+3,8.2) {\tiny $y^1_0$};
	\node at (2.9+3,8.2) {\tiny $y^k_n$};

	\node at (0.7+7.4,7.3) {\tiny $y^k_n$};
	\node at (2.9+7.4,7.3) {\tiny $y^1_0$};
	\node at (0.7+7.4,8.2) {\tiny $x^k_n$};
	\node at (2.9+7.4,8.2) {\tiny $x^1_0$};

	\node at (0.7+10.4,7.3) {\tiny $y^k_n$};
	\node at (2.9+10.4,7.3) {\tiny $y^1_0$};
	\node at (0.7+10.4,8.2) {\tiny $x^k_n$};
	\node at (2.9+10.4,8.2) {\tiny $x^1_0$};

	\draw (7.1,10.9) -- (7.1,13.5) -- (13.5,10.1) -- (7.9,10.1) -- cycle;
	\node at (9.2,11.5) {$\downarrow C^{m}$};
	
\begin{scope}[yshift=1.5cm]
	\draw (0.5,7.1) rectangle (3.1,8.4);
	\draw (3.5,7.1) rectangle (6.1,8.4);
	\draw (7.9,7.1) rectangle (10.5,8.4);
	\draw (10.9,7.1) rectangle (13.5,8.4);
	\node at (1.8,7.75) {$\uparrow MS(r_1)$};
	\node at (4.8,7.75) {$\uparrow MS(r_m)$};
	\node at (9.2,7.75) {$\downarrow RD(r_1)$};
	\node at (12.2,7.75) {$\downarrow RD(r_m)$};

	\node at (0.7,7.3) {\tiny $x^1_0$};
	\node at (2.9,7.3) {\tiny $x^k_n$};
	\node at (0.7,8.2) {\tiny $y^1_0$};
	\node at (2.9,8.2) {\tiny $y^k_n$};

	\node at (0.7+3,7.3) {\tiny $x^1_0$};
	\node at (2.9+3,7.3) {\tiny $x^k_n$};
	\node at (0.7+3,8.2) {\tiny $y^1_0$};
	\node at (2.9+3,8.2) {\tiny $y^k_n$};

	\node at (0.7+7.4,7.3) {\tiny $y^k_n$};
	\node at (2.9+7.4,7.3) {\tiny $y^1_0$};
	\node at (0.7+7.4,8.2) {\tiny $x^k_n$};
	\node at (2.9+7.4,8.2) {\tiny $x^1_0$};

	\node at (0.7+10.4,7.3) {\tiny $y^k_n$};
	\node at (2.9+10.4,7.3) {\tiny $y^1_0$};
	\node at (0.7+10.4,8.2) {\tiny $x^k_n$};
	\node at (2.9+10.4,8.2) {\tiny $x^1_0$};
\end{scope}

  \begin{scope}[yshift=3.6cm,xshift=11.8cm,rotate=90]
  	\node[leerknoten] (i10) at (0,0)  {};
		\node[knoten] (i11) at (0.2,0)  {};
		\node[knoten] (i12) at (0.4,0)  {};
		\node[knoten] (i13) at (0.6,0) {};

		\node[leerknoten] (i20) at (0.9,0) {};
		\node[knoten] (i21) at (1.1,0)  {};
		\node[knoten] (i22) at (1.3,0)  {};
		\node[knoten] (i23) at (1.5,0)  {};

		\node[leerknoten] (i30) at (1.8,0) {};
		\node[knoten] (i31) at (2,0){};
		\node[knoten] (i32) at (2.2,0)  {};
		\node[knoten] (i33) at (2.4,0)  {};
	\end{scope}
	
  \begin{scope}[yshift=3.6cm,xshift=7cm,rotate=90]
  	\node[leerknoten] (x10) at (0,0){};
		\node[knoten] (x11) at (0.2,0)  {};
		\node[knoten] (x12) at (0.4,0)  {};
		\node[knoten] (x13) at (0.6,0) {};

		\node[leerknoten] (x20) at (0.9,0) {};
		\node[knoten] (x21) at (1.1,0)  {};
		\node[knoten] (x22) at (1.3,0)  {};
		\node[knoten] (x23) at (1.5,0)  {};

		\node[leerknoten] (x30) at (1.8,0) {};
		\node[knoten] (x31) at (2,0){};
		\node[knoten] (x32) at (2.2,0)  {};
		\node[knoten] (x33) at (2.4,0)  {};
	\end{scope}
	\node at (7,3.1) {$\xnode^1_0$};
	\node at (7,6.5) {$\xnode^k_n$};
	
  \begin{scope}[yshift=13.4cm,xshift=7cm,rotate=270]
  	\node[leerknoten] (y10) at (0,0)  {};
		\node[knoten] (y11) at (0.2,0)  {};
		\node[knoten] (y12) at (0.4,0)  {};
		\node[knoten] (y13) at (0.6,0) {};

		\node[leerknoten] (y20) at (0.9,0) {};
		\node[knoten] (y21) at (1.1,0)  {};
		\node[knoten] (y22) at (1.3,0)  {};
		\node[knoten] (y23) at (1.5,0)  {};

		\node[leerknoten] (y30) at (1.8,0) {};
		\node[knoten] (y31) at (2,0){};
		\node[knoten] (y32) at (2.2,0)  {};
		\node[knoten] (y33) at (2.4,0)  {};
	\end{scope}
	\node at (7,13.9) {$\ynode^1_0$};
	\node at (7,10.5) {$\ynode^k_n$};
	
  \begin{scope}[yshift=7cm,xshift=0.6cm]
  	\node[leerknoten] (e10) at (0,0)  {};
		\node[knoten] (e11) at (0.2,0)  {};
		\node[knoten] (e12) at (0.4,0)  {};
		\node[knoten] (e13) at (0.6,0) {};

		\node[leerknoten] (e20) at (0.9,0) {};
		\node[knoten] (e21) at (1.1,0)  {};
		\node[knoten] (e22) at (1.3,0)  {};
		\node[knoten] (e23) at (1.5,0)  {};

		\node[leerknoten] (e30) at (1.8,0) {};
		\node[knoten] (e31) at (2,0){};
		\node[knoten] (e32) at (2.2,0)  {};
		\node[knoten] (e33) at (2.4,0)  {};
	\end{scope}

  \begin{scope}[yshift=7cm,xshift=3.6cm]
  	\node[leerknoten] (f10) at (0,0)  {};
		\node[knoten] (f11) at (0.2,0)  {};
		\node[knoten] (f12) at (0.4,0)  {};
		\node[knoten] (f13) at (0.6,0) {};

		\node[leerknoten] (f20) at (0.9,0) {};
		\node[knoten] (f21) at (1.1,0)  {};
		\node[knoten] (f22) at (1.3,0)  {};
		\node[knoten] (f23) at (1.5,0)  {};

		\node[leerknoten] (f30) at (1.8,0) {};
		\node[knoten] (f31) at (2,0){};
		\node[knoten] (f32) at (2.2,0)  {};
		\node[knoten] (f33) at (2.4,0)  {};
	\end{scope}
	
  \begin{scope}[yshift=7cm,xshift=10.4cm,rotate=180]
  	\node[leerknoten] (g10) at (0,0)  {};
		\node[knoten] (g11) at (0.2,0)  {};
		\node[knoten] (g12) at (0.4,0)  {};
		\node[knoten] (g13) at (0.6,0) {};

		\node[leerknoten] (g20) at (0.9,0) {};
		\node[knoten] (g21) at (1.1,0)  {};
		\node[knoten] (g22) at (1.3,0)  {};
		\node[knoten] (g23) at (1.5,0)  {};

		\node[leerknoten] (g30) at (1.8,0) {};
		\node[knoten] (g31) at (2,0){};
		\node[knoten] (g32) at (2.2,0)  {};
		\node[knoten] (g33) at (2.4,0)  {};
	\end{scope}
	
  \begin{scope}[yshift=7cm,xshift=13.4cm,rotate=180]
  	\node[leerknoten] (h10) at (0,0)  {};
		\node[knoten] (h11) at (0.2,0)  {};
		\node[knoten] (h12) at (0.4,0)  {};
		\node[knoten] (h13) at (0.6,0) {};

		\node[leerknoten] (h20) at (0.9,0) {};
		\node[knoten] (h21) at (1.1,0)  {};
		\node[knoten] (h22) at (1.3,0)  {};
		\node[knoten] (h23) at (1.5,0)  {};

		\node[leerknoten] (h30) at (1.8,0) {};
		\node[knoten] (h31) at (2,0){};
		\node[knoten] (h32) at (2.2,0)  {};
		\node[knoten] (h33) at (2.4,0)  {};
	\end{scope}

\begin{scope}[yshift=1.5cm]
  \begin{scope}[yshift=7cm,xshift=0.6cm]
  	\node[leerknoten] (r10) at (0,0)  {};
		\node[knoten] (r11) at (0.2,0)  {};
		\node[knoten] (r12) at (0.4,0)  {};
		\node[knoten] (r13) at (0.6,0) {};

		\node[leerknoten] (r20) at (0.9,0) {};
		\node[knoten] (r21) at (1.1,0)  {};
		\node[knoten] (r22) at (1.3,0)  {};
		\node[knoten] (r23) at (1.5,0)  {};

		\node[leerknoten] (r30) at (1.8,0) {};
		\node[knoten] (r31) at (2,0){};
		\node[knoten] (r32) at (2.2,0)  {};
		\node[knoten] (r33) at (2.4,0)  {};
	\end{scope}

  \begin{scope}[yshift=7cm,xshift=3.6cm]
  	\node[leerknoten] (t10) at (0,0)  {};
		\node[knoten] (t11) at (0.2,0)  {};
		\node[knoten] (t12) at (0.4,0)  {};
		\node[knoten] (t13) at (0.6,0) {};

		\node[leerknoten] (t20) at (0.9,0) {};
		\node[knoten] (t21) at (1.1,0)  {};
		\node[knoten] (t22) at (1.3,0)  {};
		\node[knoten] (t23) at (1.5,0)  {};

		\node[leerknoten] (t30) at (1.8,0) {};
		\node[knoten] (t31) at (2,0){};
		\node[knoten] (t32) at (2.2,0)  {};
		\node[knoten] (t33) at (2.4,0)  {};
	\end{scope}
	
  \begin{scope}[yshift=7cm,xshift=10.4cm,rotate=180]
  	\node[leerknoten] (z10) at (0,0)  {};
		\node[knoten] (z11) at (0.2,0)  {};
		\node[knoten] (z12) at (0.4,0)  {};
		\node[knoten] (z13) at (0.6,0) {};

		\node[leerknoten] (z20) at (0.9,0) {};
		\node[knoten] (z21) at (1.1,0)  {};
		\node[knoten] (z22) at (1.3,0)  {};
		\node[knoten] (z23) at (1.5,0)  {};

		\node[leerknoten] (z30) at (1.8,0) {};
		\node[knoten] (z31) at (2,0){};
		\node[knoten] (z32) at (2.2,0)  {};
		\node[knoten] (z33) at (2.4,0)  {};
	\end{scope}
	
  \begin{scope}[yshift=7cm,xshift=13.4cm,rotate=180]
  	\node[leerknoten] (u10) at (0,0)  {};
		\node[knoten] (u11) at (0.2,0)  {};
		\node[knoten] (u12) at (0.4,0)  {};
		\node[knoten] (u13) at (0.6,0) {};

		\node[leerknoten] (u20) at (0.9,0) {};
		\node[knoten] (u21) at (1.1,0)  {};
		\node[knoten] (u22) at (1.3,0)  {};
		\node[knoten] (u23) at (1.5,0)  {};

		\node[leerknoten] (u30) at (1.8,0) {};
		\node[knoten] (u31) at (2,0){};
		\node[knoten] (u32) at (2.2,0)  {};
		\node[knoten] (u33) at (2.4,0)  {};
	\end{scope}
\end{scope}

\begin{scope}[yshift=3cm]
  \begin{scope}[yshift=7cm,xshift=0.6cm]
  	\node[leerknoten] (a10) at (0,0)  {};
		\node[knoten] (a11) at (0.2,0)  {};
		\node[knoten] (a12) at (0.4,0)  {};
		\node[knoten] (a13) at (0.6,0) {};

		\node[leerknoten] (a20) at (0.9,0) {};
		\node[knoten] (a21) at (1.1,0)  {};
		\node[knoten] (a22) at (1.3,0)  {};
		\node[knoten] (a23) at (1.5,0)  {};

		\node[leerknoten] (a30) at (1.8,0) {};
		\node[knoten] (a31) at (2,0){};
		\node[knoten] (a32) at (2.2,0)  {};
		\node[knoten] (a33) at (2.4,0)  {};
	\end{scope}

  \begin{scope}[yshift=7cm,xshift=3.6cm]
  	\node[leerknoten] (b10) at (0,0)  {};
		\node[knoten] (b11) at (0.2,0)  {};
		\node[knoten] (b12) at (0.4,0)  {};
		\node[knoten] (b13) at (0.6,0) {};

		\node[leerknoten] (b20) at (0.9,0) {};
		\node[knoten] (b21) at (1.1,0)  {};
		\node[knoten] (b22) at (1.3,0)  {};
		\node[knoten] (b23) at (1.5,0)  {};

		\node[leerknoten] (b30) at (1.8,0) {};
		\node[knoten] (b31) at (2,0){};
		\node[knoten] (b32) at (2.2,0)  {};
		\node[knoten] (b33) at (2.4,0)  {};
	\end{scope}
	
  \begin{scope}[yshift=7cm,xshift=10.4cm,rotate=180]
  	\node[leerknoten] (c10) at (0,0)  {};
		\node[knoten] (c11) at (0.2,0)  {};
		\node[knoten] (c12) at (0.4,0)  {};
		\node[knoten] (c13) at (0.6,0) {};

		\node[leerknoten] (c20) at (0.9,0) {};
		\node[knoten] (c21) at (1.1,0)  {};
		\node[knoten] (c22) at (1.3,0)  {};
		\node[knoten] (c23) at (1.5,0)  {};

		\node[leerknoten] (c30) at (1.8,0) {};
		\node[knoten] (c31) at (2,0){};
		\node[knoten] (c32) at (2.2,0)  {};
		\node[knoten] (c33) at (2.4,0)  {};
	\end{scope}
	
  \begin{scope}[yshift=7cm,xshift=13.4cm,rotate=180]
  	\node[leerknoten] (d10) at (0,0)  {};
		\node[knoten] (d11) at (0.2,0)  {};
		\node[knoten] (d12) at (0.4,0)  {};
		\node[knoten] (d13) at (0.6,0) {};

		\node[leerknoten] (d20) at (0.9,0) {};
		\node[knoten] (d21) at (1.1,0)  {};
		\node[knoten] (d22) at (1.3,0)  {};
		\node[knoten] (d23) at (1.5,0)  {};

		\node[leerknoten] (d30) at (1.8,0) {};
		\node[knoten] (d31) at (2,0){};
		\node[knoten] (d32) at (2.2,0)  {};
		\node[knoten] (d33) at (2.4,0)  {};
	\end{scope}
\end{scope}
	
\foreach \x in {1,2,3} \foreach \y in {0,1,2,3} {
	\draw[densely dotted] (i\x\y) -- (x\x\y);
	\draw[densely dotted] (e\x\y) -- (x\x\y);
	\draw[densely dotted] (f\x\y) -- (x\x\y);
	\draw[densely dotted] (g\x\y) -- (x\x\y);
	\draw[densely dotted] (h\x\y) -- (x\x\y);
	\draw[densely dotted] (a\x\y) -- (y\x\y);
	\draw[densely dotted] (b\x\y) -- (y\x\y);
}

\end{tikzpicture}

%% file: Tikz/Tikz-RS-LICS.tex
\begin{tikzpicture}
	[scale=1, transform shape, knoten/.style={circle,draw=black,fill=black,
	inner  sep=0pt,minimum  size=1mm},kknoten/.style={circle,draw=black,fill=black,
	inner  sep=0pt,minimum  size=0.8mm}, leerknoten/.style={circle,draw=black,fill=white,
	inner  sep=0pt,minimum  size=1mm},leerkknoten/.style={circle,draw=black,fill=white,
	inner  sep=0pt,minimum  size=0.8mm}]
	
	\node[knoten] (xc) at (3,0) [label=above:{{\tiny$x^c$}}] {};
	\node[knoten] (xd) at (3.5,0) [label=above:{{\tiny$x^d$}}] {};
	\node[knoten] (xi) at (4,0) [label=above:{{\tiny$x^i$}}] {};
	
	\begin{scope}[yshift=-1.5cm]

	\node[knoten] (yc) at (3,0) [label=below:{{\tiny$y^c$}}] {};
	\node[knoten] (yd) at (3.5,0) [label=below:{{\tiny$y^d$}}] {};
	\node[knoten] (yi) at (4,0) [label=below:{{\tiny$y^i$}}] {};
	
	\end{scope}
	
		\draw[-] (xc) -- (yc);
	\draw[-] (xd) -- (yd);
	\draw[-] (xi) -- (yi);

	\node[leerknoten] (x10) at (5,0) [label=above:{{\tiny$x^c_0$}}] {};
	\node[knoten] (x11) at (5.3,0) [label=above:{{\tiny$x^c_1$}}] {};
	\node[knoten] (x12) at (5.6,0) [label=above:{{\tiny$x^c_u$}}] {};
	\node[knoten] (x13) at (5.9,0) {};
	\node[knoten] (x14) at (6.2,0) {};
	\node[knoten] (x15) at (6.5,0) [label=above:{{\tiny$x^c_n$}}] {};

	\node[leerknoten] (x20) at (5+2,0) [label=above:{{\tiny$x^d_0$}}] {};
	\node[knoten] (x21) at (5.3+2,0) [label=above:{{\tiny$x^d_1$}}] {};
	\node[knoten] (x22) at (5.6+2,0)  {};
	\node[knoten] (x23) at (5.9+2,0) [label=above:{{\tiny$x^d_v$}}] {};
	\node[knoten] (x24) at (6.2+2,0)  {};
	\node[knoten] (x25) at (6.5+2,0) [label=above:{{\tiny$x^d_n$}}] {};

	\node[leerknoten] (x30) at (5+4,0) [label=above:{{\tiny$x^i_0$}}] {};
	\node[knoten] (x31) at (5.3+4,0) [label=above:{{\tiny$x^i_1$}}] {};
	\node[knoten] (x32) at (5.6+4,0)  {};
	\node[knoten] (x33) at (5.9+4,0)  {};
	\node[knoten] (x34) at (6.2+4,0) [label=above:{{\tiny$x^i_w$}}] {};
	\node[knoten] (x35) at (6.5+4,0) [label=above:{{\tiny$x^i_n$}}] {};

\begin{scope}[yshift=-1.5cm]

	\node[leerknoten] (y10) at (5,0) [label=below:{{\tiny$y^c_0$}}] {};
	\node[knoten] (y11) at (5.3,0) [label=below:{{\tiny$y^c_1$}}] {};
	\node[knoten] (y12) at (5.6,0)  {};
	\node[knoten] (y13) at (5.9,0) {};
	\node[knoten] (y14) at (6.2,0) [label=below:{{\tiny$y^c_w$}}] {};
	\node[knoten] (y15) at (6.5,0) [label=below:{{\tiny$y^c_n$}}] {};

	\node[leerknoten] (y20) at (5+2,0) [label=below:{{\tiny$y^d_0$}}] {};
	\node[knoten] (y21) at (5.3+2,0) [label=below:{{\tiny$y^d_1$}}] {};
	\node[knoten] (y22) at (5.6+2,0)  {};
	\node[knoten] (y23) at (5.9+2,0) [label=below:{{\tiny$y^d_v$}}] {};
	\node[knoten] (y24) at (6.2+2,0)  {};
	\node[knoten] (y25) at (6.5+2,0) [label=below:{{\tiny$y^d_n$}}] {};

	\node[leerknoten] (y30) at (5+4,0) [label=below:{{\tiny$y^i_0$}}] {};
	\node[knoten] (y31) at (5.3+4,0) [label=below:{{\tiny$y^i_1$}}] {};
	\node[knoten] (y32) at (5.6+4,0)  {};
	\node[knoten] (y33) at (5.9+4,0)  {};
	\node[knoten] (y34) at (6.2+4,0) [label=below:{{\tiny$y^i_w$}}] {};
	\node[knoten] (y35) at (6.5+4,0) [label=below:{{\tiny$y^i_n$}}] {};

\end{scope}

	\foreach \x in {0,1,3,4,5} \draw[-] (x1\x) -- (y10);
	\draw[-] (x12) -- (y14);

	\foreach \x in {0,1,2,4,5} \draw[-] (x2\x) -- (y20);
	\draw[-] (x23) -- (y23);

	\foreach \x in {0,1,2,3,5} \draw[-] (x3\x) -- (y3\x);
	\draw[-] (x34) -- (y30);

\end{tikzpicture}

%% file: Tikz/Tikz-RD-LICS.tex
\begin{tikzpicture}
	[scale=1, transform shape, knoten/.style={circle,draw=black,fill=black,
	inner  sep=0pt,minimum  size=1mm},kknoten/.style={circle,draw=black,fill=black,
	inner  sep=0pt,minimum  size=0.8mm}, leerknoten/.style={circle,draw=black,fill=white,
	inner  sep=0pt,minimum  size=1mm},leerkknoten/.style={circle,draw=black,fill=white,
	inner  sep=0pt,minimum  size=0.8mm}]
	
	\node[knoten] (xc) at (3,0) [label=above:{{\tiny$x^c$}}] {};
	\node[knoten] (xd) at (3.5,0) [label=above:{{\tiny$x^d$}}] {};
	\node[knoten] (xi) at (4,0) [label=above:{{\tiny$x^i$}}] {};

	\node[leerknoten] (x10) at (5,0) [label=above:{{\tiny$x^c_0$}}] {};
	\node[knoten] (x11) at (5.3,0) [label=above:{{\tiny$x^c_1$}}] {};
	\node[knoten] (x12) at (5.6,0) [label=above:{{\tiny$x^c_u$}}] {};
	\node[knoten] (x13) at (5.9,0) {};
	\node[knoten] (x14) at (6.2,0) {};
	\node[knoten] (x15) at (6.5,0) [label=above:{{\tiny$x^c_n$}}] {};

	\node[leerknoten] (x20) at (5+2,0) [label=above:{{\tiny$x^d_0$}}] {};
	\node[knoten] (x21) at (5.3+2,0) [label=above:{{\tiny$x^d_1$}}] {};
	\node[knoten] (x22) at (5.6+2,0)  {};
	\node[knoten] (x23) at (5.9+2,0) [label=above:{{\tiny$x^d_v$}}] {};
	\node[knoten] (x24) at (6.2+2,0)  {};
	\node[knoten] (x25) at (6.5+2,0) [label=above:{{\tiny$x^d_n$}}] {};

	\node[leerknoten] (x30) at (5+4,0) [label=above:{{\tiny$x^i_0$}}] {};
	\node[knoten] (x31) at (5.3+4,0) [label=above:{{\tiny$x^i_1$}}] {};
	\node[knoten] (x32) at (5.6+4,0)  {};
	\node[knoten] (x33) at (5.9+4,0)  {};
	\node[knoten] (x34) at (6.2+4,0) [label=above:{{\tiny$x^i_w$}}] {};
	\node[knoten] (x35) at (6.5+4,0) [label=above:{{\tiny$x^i_n$}}] {};

\begin{scope}[yshift=-1.5cm]

	\node[knoten] (yc) at (3,0) [label=below:{{\tiny$y^c$}}] {};
	\node[knoten] (yd) at (3.5,0) [label=below:{{\tiny$y^d$}}] {};
	\node[knoten] (yi) at (4,0) [label=below:{{\tiny$y^i$}}] {};

	\node[leerknoten] (y10) at (5,0) [label=below:{{\tiny$y^c_0$}}] {};
	\node[knoten] (y11) at (5.3,0) [label=below:{{\tiny$y^c_1$}}] {};
	\node[knoten] (y12) at (5.6,0)  {};
	\node[knoten] (y13) at (5.9,0) {};
	\node[knoten] (y14) at (6.2,0) [label=below:{{\tiny$y^c_w$}}] {};
	\node[knoten] (y15) at (6.5,0) [label=below:{{\tiny$y^c_n$}}] {};

	\node[leerknoten] (y20) at (5+2,0) [label=below:{{\tiny$y^d_0$}}] {};
	\node[knoten] (y21) at (5.3+2,0) [label=below:{{\tiny$y^d_1$}}] {};
	\node[knoten] (y22) at (5.6+2,0)  {};
	\node[knoten] (y23) at (5.9+2,0) [label=below:{{\tiny$y^d_v$}}] {};
	\node[knoten] (y24) at (6.2+2,0)  {};
	\node[knoten] (y25) at (6.5+2,0) [label=below:{{\tiny$y^d_n$}}] {};

	\node[leerknoten] (y30) at (5+4,0) [label=below:{{\tiny$y^i_0$}}] {};
	\node[knoten] (y31) at (5.3+4,0) [label=below:{{\tiny$y^i_1$}}] {};
	\node[knoten] (y32) at (5.6+4,0)  {};
	\node[knoten] (y33) at (5.9+4,0)  {};
	\node[knoten] (y34) at (6.2+4,0) [label=below:{{\tiny$y^i_w$}}] {};
	\node[knoten] (y35) at (6.5+4,0) [label=below:{{\tiny$y^i_n$}}] {};

\end{scope}

	\draw[-] (xc) -- (yc);
	\draw[-] (xd) -- (yd);
	\draw[-] (xi) -- (yi);

	\draw[-] (x10) -- (y10);
	\draw[-] (x12) -- (y14);

	\draw[-] (x20) -- (y20);
	\draw[-] (x23) -- (y23);

	\foreach \x in {0,1,2,3,5} \draw[-] (x3\x) -- (y3\x);

\end{tikzpicture}

%% file: Tikz/Tikz-multi-LICS.tex
\begin{tikzpicture}
			[scale=1, transform shape, knoten/.style={circle,draw=black,fill=black,
			inner  sep=0pt,minimum  size=1mm},kknoten/.style={circle,draw=black,fill=black,
			inner  sep=0pt,minimum  size=0.8mm}, leerknoten/.style={circle,draw=black,fill=white,
			inner  sep=0pt,minimum  size=1mm},leerkknoten/.style={circle,draw=black,fill=white,
			inner  sep=0pt,minimum  size=0.8mm}]
			
			
			
			\node at (1,-6) {Spoiler's side $M^{k,n}_S$};
			\node at (11,-6) {Duplicator's side $M^{k,n}_D$};

			
			\node[knoten] (x1) at (0,0) [label=above:{{\tiny$x^1$}}] {};
			\node[knoten] (x2) at (1,0)  {};
			\node[knoten] (x3) at (2,0) [label=above:{{\tiny$x^k$}}] {};
			
			\node[knoten] (a1) at (0,-1.5) [label=left:{{\tiny$a^1$}}] {};
			\node[knoten] (a2) at (1,-1.5) {};
			\node[knoten] (a3) at (2,-1.5) [label=right:{{\tiny$a^k$}}] {};

			\node[knoten] (b1) at (0,-3) [label=left:{{\tiny$b^1$}}] {};
			\node[knoten] (b2) at (1,-3) {};
			\node[knoten] (b3) at (2,-3) [label=right:{{\tiny$b^k$}}] {};

			\node[knoten] (y1) at (0,-4.5) [label=below:{{\tiny$y^1$}}] {};
			\node[knoten] (y2) at (1,-4.5) {};
			\node[knoten] (y3) at (2,-4.5) [label=below:{{\tiny$y^k$}}] {};

			\foreach \i in {1,2,3}{			
			\draw[-] (x\i) -- (a\i);
			\draw[-] (b\i) -- (y\i);
			}
			
			\foreach \i in {1,2,3} {\foreach \j in {1,2,3}{			
			\draw[-] (a\i) -- (b\j);
			}}
			
			\foreach \i/\j in {1/2,2/3}{			
			\draw[-] (a\i) -- (a\j);
			\draw[-] (b\i) -- (b\j);
			}
			
			\draw[-] (a1) .. controls +(1,-0.5) .. (a3);
			\draw[-] (b1) .. controls +(1,0.5) .. (b3);

			\begin{scope}[xshift=5cm]
			
			\foreach \weite/\tiefe/\schrift/\ttiefe/\kante in {0.3/1.5/0.5/3/very thin} {
			
			\node[leerknoten] (x10) at (0,0) [label=above:{{\tiny$x^1_0$}}] {};
			\node[leerknoten] (x20) at (5.5,0) [label=above:{{\tiny$x^2_0$}}] {};
			\node[leerknoten] (x30) at (11,0) [label=above:{{\tiny$x^k_0$}}] {};
			
			\node[knoten] (x11) at (1*1,0)[label=above:{\tiny$x^1_1$}] {};		
			\node[knoten] (x12) at (2*1,0)[label=above:{\tiny$x^1_2$}] {};	
			\node[knoten] (x13) at (3*1,0) {};	
			\node[knoten] (x14) at (4*1,0)[label=above:{\tiny$x^1_n$}] {};		
			
			\node[knoten] (x21) at (5.5+1,0)[label=above:{\tiny$x^2_1$}] {};		
			\node[knoten] (x22) at (5.5+2,0)[label=above:{\tiny$x^2_2$}] {};	
			\node[knoten] (x23) at (5.5+3,0) {};	
			\node[knoten] (x24) at (5.5+4,0)[label=above:{\tiny$x^2_n$}] {};	
			
			\node[knoten] (x31) at (11+1*0.2,0) {};		
			\node[knoten] (x32) at (11+2*0.2,0) {};	
			\node[knoten] (x33) at (11+3*0.2,0) {};	
			\node[knoten] (x34) at (11+4*0.2,0)[label=above:{\tiny$x^k_n$}] {};

			\node[kknoten] (a111) at (1-\weite,-\tiefe){};
			\node[kknoten] (a112) at (1,-\tiefe) {};
			\node[kknoten] (a113) at (1+\weite,-\tiefe) {};

			\node[kknoten] (a121) at (2-\weite,-\tiefe) {};
			\node[kknoten] (a122) at (2,-\tiefe) {};
			\node[kknoten] (a123) at (2+\weite,-\tiefe) {};

			\node[kknoten] (a131) at (3-\weite,-\tiefe) {};
			\node[kknoten] (a132) at (3,-\tiefe) {};
			\node[kknoten] (a133) at (3+\weite,-\tiefe) {};

			\node[kknoten] (a141) at (4-\weite,-\tiefe) {};
			\node[kknoten] (a142) at (4,-\tiefe) {};
			\node[kknoten] (a143) at (4+\weite,-\tiefe) {};

			\node[kknoten] (a211) at (5.5+1-\weite,-\tiefe) {};
			\node[kknoten] (a212) at (5.5+1,-\tiefe) {};
			\node[kknoten] (a213) at (5.5+1+\weite,-\tiefe) {};

			\node[kknoten] (a221) at (5.5+2-\weite,-\tiefe) {};
			\node[kknoten] (a222) at (5.5+2,-\tiefe) {};
			\node[kknoten] (a223) at (5.5+2+\weite,-\tiefe) {};

			\node[kknoten] (a231) at (5.5+3-\weite,-\tiefe) {};
			\node[kknoten] (a232) at (5.5+3,-\tiefe) {};
			\node[kknoten] (a233) at (5.5+3+\weite,-\tiefe) {};

			\node[kknoten] (a241) at (5.5+4-\weite,-\tiefe) {};
			\node[kknoten] (a242) at (5.5+4,-\tiefe) {};
			\node[kknoten] (a243) at (5.5+4+\weite,-\tiefe) {};

			\node[leerkknoten] (a10) at (0,-\tiefe) {};
			\node[leerkknoten] (a20) at (5.5,-\tiefe) {};
			
			\foreach \l in {1,2} { \foreach \n in {1,2,3,4} { \foreach \k in {1,2,3} {
			  \draw[-,\kante] (x\l\n) -- (a\l\n\k);
			}}}
			
			\draw[-,\kante] (x10) -- (a10);
			\draw[-,\kante] (x20) -- (a20);

			\node[leerkknoten] (b101) at (0-\weite,-\ttiefe) {};
			\node[leerkknoten] (b102) at (0,-\ttiefe) {};
			\node[leerkknoten] (b103) at (0+\weite,-\ttiefe) {};

			\node[kknoten] (b11) at (1,-\ttiefe) {};
			\node[kknoten] (b12) at (2,-\ttiefe) {};
			\node[kknoten] (b13) at (3,-\ttiefe) {};
			\node[kknoten] (b14) at (4,-\ttiefe) {};

			\node[leerkknoten] (b201) at (5.5+0-\weite,-\ttiefe) {};
			\node[leerkknoten] (b202) at (5.5+0,-\ttiefe) {};
			\node[leerkknoten] (b203) at (5.5+0+\weite,-\ttiefe) {};

			\node[kknoten] (b21) at (5.5+1,-\ttiefe) {};
			\node[kknoten] (b22) at (5.5+2,-\ttiefe) {};
			\node[kknoten] (b23) at (5.5+3,-\ttiefe) {};
			\node[kknoten] (b24) at (5.5+4,-\ttiefe) {};

			\foreach \l in {1,2} { \foreach \n in {1,2,3,4} { \foreach \k in {1,2,3} {
			  \draw[-,\kante] (b\l\n) -- (a\l\n\k);
			}}}
			
			\draw[-,\kante] (b101) -- (a10);
			\draw[-,\kante] (b201) -- (a20);
			\draw[-,\kante] (b102) -- (a10);
			\draw[-,\kante] (b202) -- (a20);
			\draw[-,\kante] (b103) -- (a10);
			\draw[-,\kante] (b203) -- (a20);

			\begin{scope}[yshift=-4.5cm]
			
			\node[leerknoten] (y10) at (0,0) [label=below:{{\tiny$y^1_0$}}] {};
			\node[leerknoten] (y20) at (5.5,0) [label=below:{{\tiny$y^2_0$}}] {};
			\node[leerknoten] (y30) at (11,0) [label=below:{{\tiny$y^k_0$}}] {};
			
			\node[knoten] (y11) at (1*1,0)[label=below:{\tiny$y^1_1$}] {};		
			\node[knoten] (y12) at (2*1,0)[label=below:{\tiny$y^1_2$}] {};	
			\node[knoten] (y13) at (3*1,0) {};	
			\node[knoten] (y14) at (4*1,0)[label=below:{\tiny$y^1_n$}] {};		
			
			\node[knoten] (y21) at (5.5+1,0)[label=below:{\tiny$y^2_1$}] {};		
			\node[knoten] (y22) at (5.5+2,0)[label=below:{\tiny$y^2_2$}] {};	
			\node[knoten] (y23) at (5.5+3,0) {};	
			\node[knoten] (y24) at (5.5+4,0)[label=below:{\tiny$y^2_n$}] {};	
			
			\node[knoten] (y31) at (11+1*0.2,0) {};		
			\node[knoten] (y32) at (11+2*0.2,0) {};	
			\node[knoten] (y33) at (11+3*0.2,0) {};	
			\node[knoten] (y34) at (11+4*0.2,0)[label=below:{\tiny$y^k_n$}] {};	

			\end{scope}

			\foreach \l in {1,2} {  \foreach \n in {1,2,3,4} {
			  \draw[-,\kante] (b\l\n) -- (y\l\n);
			}}

			\foreach \l in {1,2} {  \foreach \k in {1,2,3} {
			  \draw[-,\kante] (b\l0\k) -- (y\l0);
			}}

			\foreach \l in {1,2} {  \foreach \n in {1,2,3,4} { \foreach \k in {1,2,3} {
			  \draw[-,\kante] (a\l\n\k) -- (x\l0);
			}}}

			\foreach \l in {1,2} {  \foreach \n in {1,2,3,4} { 
			  \draw[-,\kante] (a\l0) -- (b\l\n);
			}}

			\foreach \l in {1,2} {  \foreach \n in {1,2,3,4} { \foreach \k in {1,2,3} {
			  \draw[-,\kante] (b\l0\k) -- (y\l\n);
			}}}

			\foreach  \l in {1} { \foreach \n in {1,2,3,4} { \foreach \k/\kk in {1/2,2/1,1/3,3/1,2/3,3/2} {
			  \draw[-,\kante] (b\l0\k) -- (a\l\n\kk);
			}}}

			
			\node at (1-\weite-0.18,-\tiefe) {\scalebox{\schrift}{$a^1_{1,1}$}};
			\node at (4+\weite+0.35,-\tiefe) {\tiny{$a^1_{n,k}$}};
			\node at (5.5+1-\weite-0.18,-\tiefe) {\scalebox{\schrift}{$a^2_{1,1}$}};
			\node at (5.5+4+\weite+0.35,-\tiefe) {\tiny{$a^2_{n,k}$}};
			\node at (0-0.20,-\tiefe) {\tiny{$a^1_{0}$}};
			\node at (5.5-0.20,-\tiefe) {\tiny{$a^2_{0}$}};
			\node at (0-\weite-0.25,-\ttiefe) {\tiny{$b^1_{0,1}$}};
			\node at (1+0.2,-\ttiefe) {\scalebox{\schrift}{$b^1_{1}$}};
			\node at (4+0.3,-\ttiefe) {\tiny{$b^1_{n}$}};
			\node at (5.5+0-\weite-0.25,-\ttiefe) {\tiny{$b^2_{0,1}$}};
			\node at (5.5+0+\weite+0.18,-\ttiefe+0.1) {\scalebox{\schrift}{$b^2_{0,k}$}};
			\node at (5.5+1+0.2,-\ttiefe) {\scalebox{\schrift}{$b^2_{1}$}};
			\node at (5.5+4+0.3,-\ttiefe) {\tiny{$b^2_{n}$}};

 			}
			\end{scope}
\end{tikzpicture}

%% file: Tikz/Tikz-INIT.tex
\newcommand{\schriftgr}{\small}

\begin{tikzpicture}
	[scale=0.755, transform shape, knoten/.style={circle,draw=black,fill=black,
	inner  sep=0pt,minimum  size=1mm},kknoten/.style={circle,draw=black,fill=black,
	inner  sep=0pt,minimum  size=0.8mm}, leerknoten/.style={circle,draw=black,fill=white,
	inner  sep=0pt,minimum  size=1mm},leerkknoten/.style={circle,draw=black,fill=white,
	inner  sep=0pt,minimum  size=0.8mm}]
	
	\node at (1.6+2,-5) {Duplicator's side};

  \node[knoten] (y1) at (3.2-0.5,2) [label=above:{\schriftgr{$y_{1}$}}] {};
  \node[knoten] (y2) at (4+0.5,2) [label=above:{\schriftgr{$y_{2}$}}] {};

	\node[leerknoten] (a10) at (0,0)  [label=below:{\schriftgr{$a^{1}_{0}$}}]{};
	\node[knoten] (a11) at (0.2,0)  {};
	\node[knoten] (a12) at (0.4,0)   {};
	\node[knoten] (a13) at (0.6,0) {};
	\node[knoten] (a14) at (0.8,0) [label=below:{\schriftgr{$a^{1}_{n}$}}]{};

	\node[leerknoten] (a20) at (1.2,0) {};
	\node[knoten] (a21) at (1.4,0)  {};
	\node[knoten] (a22) at (1.6,0)  {};
	\node[knoten] (a23) at (1.8,0)  {};
	\node[knoten] (a24) at (2,0) {};

	\node[leerknoten] (a30) at (2.4,0) {};
	\node[knoten] (a31) at (2.6,0){};
	\node[knoten] (a32) at (2.8,0)  {};
	\node[knoten] (a33) at (3,0)  {};
	\node[knoten] (a34) at (3.2,0) [label=below:{\schriftgr{$a^{k}_{n}$}}] {};
	
	\draw (-0.2,-0.1) rectangle (3.4,-1.9);
	\node at (1.6,-1) {$M_D^{1}$};
	
  \begin{scope}[yshift=-2cm]
  		\node[leerknoten] (c10) at (0,0) [label=above:{\schriftgr{$c^{1}_0$}}] {};
		\node[knoten] (c11) at (0.2,0)  {};
		\node[knoten] (c12) at (0.4,0)  {};
		\node[knoten] (c13) at (0.6,0) {};
		\node[knoten] (c14) at (0.8,0)[label=above:{\schriftgr{$c^{1}_n$}}] {};

		\node[leerknoten] (c20) at (1.2,0) {};
		\node[knoten] (c21) at (1.4,0)  {};
		\node[knoten] (c22) at (1.6,0)  {};
		\node[knoten] (c23) at (1.8,0)  {};
		\node[knoten] (c24) at (2,0)  {};

		\node[leerknoten] (c30) at (2.4,0) {};
		\node[knoten] (c31) at (2.6,0){};
		\node[knoten] (c32) at (2.8,0)  {};
		\node[knoten] (c33) at (3,0)  {};
		\node[knoten] (c34) at (3.2,0) [label=above:{\schriftgr{$c^{k}_n$}}] {};
	\end{scope}

\begin{scope}[xshift=4cm]
	
	\node[leerknoten] (b10) at (0,0) [label=below:{\schriftgr{$b^{1}_0$}}] {};
	\node[knoten] (b11) at (0.2,0)  {};
	\node[knoten] (b12) at (0.4,0)  {};
	\node[knoten] (b13) at (0.6,0) {};
	\node[knoten] (b14) at (0.8,0) [label=below:{\schriftgr{$b^{1}_n$}}] {};

	\node[leerknoten] (b20) at (1.2,0) {};
	\node[knoten] (b21) at (1.4,0)  {};
	\node[knoten] (b22) at (1.6,0)  {};
	\node[knoten] (b23) at (1.8,0)  {};
	\node[knoten] (b24) at (2,0)  {};

	\node[leerknoten] (b30) at (2.4,0) {};
	\node[knoten] (b31) at (2.6,0){};
	\node[knoten] (b32) at (2.8,0)  {};
	\node[knoten] (b33) at (3,0)  {};
	\node[knoten] (b34) at (3.2,0) [label=below:{\schriftgr{$b^{k}_n$}}] {};
	
	\draw (-0.2,-0.1) rectangle (3.4,-1.9);
	\node at (1.6,-1) {$M_D^{2}$};
	
  \begin{scope}[yshift=-2cm]
  	\node[leerknoten] (d10) at (0,0) [label=above:{\schriftgr{$d^{1}_0$}}] {};
		\node[knoten] (d11) at (0.2,0)  {};
		\node[knoten] (d12) at (0.4,0)  {};
		\node[knoten] (d13) at (0.6,0) {};
		\node[knoten] (d14) at (0.8,0) [label=above:{\schriftgr{$d^{1}_n$}}] {};

		\node[leerknoten] (d20) at (1.2,0) {};
		\node[knoten] (d21) at (1.4,0)  {};
		\node[knoten] (d22) at (1.6,0)  {};
		\node[knoten] (d23) at (1.8,0)  {};
		\node[knoten] (d24) at (2,0)  {};

		\node[leerknoten] (d30) at (2.4,0) {};
		\node[knoten] (d31) at (2.6,0){};
		\node[knoten] (d32) at (2.8,0)  {};
		\node[knoten] (d33) at (3,0)  {};
		\node[knoten] (d34) at (3.2,0) [label=above:{\schriftgr{$d^{k}_n$}}] {};
	\end{scope}
\end{scope}

\begin{scope}[yshift=-4cm,xshift=2cm]
 	\node[leerknoten] (x10) at (0,0) [label=below:{\schriftgr{$x^{1}_0$}}] {};
	\node[knoten] (x11) at (0.2,0)  {};
	\node[knoten] (x12) at (0.4,0)  {};
	\node[knoten] (x13) at (0.6,0) {};
	\node[knoten] (x14) at (0.8,0) [label=below:{\schriftgr{$x^{1}_n$}}] {};

	\node[leerknoten] (x20) at (1.2,0) {};
	\node[knoten] (x21) at (1.4,0)  {};
	\node[knoten] (x22) at (1.6,0)  {};
	\node[knoten] (x23) at (1.8,0)  {};
	\node[knoten] (x24) at (2,0)  {};

	\node[leerknoten] (x30) at (2.4,0) {};
	\node[knoten] (x31) at (2.6,0){};
	\node[knoten] (x32) at (2.8,0)  {};
	\node[knoten] (x33) at (3,0)  {};
	\node[knoten] (x34) at (3.2,0) [label=below:{\schriftgr{$x^{k}_n$}}] {};
\end{scope}
	

\draw[-] (y2) -- (a10);
\draw[-] (y2) -- (a20);
\draw[-] (y2) -- (a30);

\draw[-] (y2) -- (b12);
\draw[-] (y2) -- (b24);
\draw[-] (y2) -- (b33);

\draw[-] (y2) -- (a12);
\draw[-] (y2) -- (a24);
\draw[-] (y2) -- (a33);

\draw[-] (y1) -- (b10);
\draw[-] (y1) -- (b20);
\draw[-] (y1) -- (b30);

\draw[-] (y1) -- (a12);
\draw[-] (y1) -- (a24);
\draw[-] (y1) -- (a33);

\draw[-] (y1) -- (b12);
\draw[-] (y1) -- (b24);
\draw[-] (y1) -- (b33);

\foreach \i/\j in {1/2,2/4,3/3} {
	\draw[-] (c\i\j) -- (x\i\j);
	\draw[-] (d\i\j) -- (x\i\j);
};

\foreach \i in {1,2,3}\foreach \l in {0,1,2,3,4} {
	\draw[-] (c\i0) -- (x\i\l);
	\draw[-] (d\i0) -- (x\i\l);
};


\begin{scope}[xshift=-6cm]
	\node at (1.4,-5) {Spoiler's side};

	\node[knoten] (y) at (1.4,2) [label=above:{\schriftgr{$y$}}]{};
	\node[knoten] (x1) at (1.4-0.5,-4)[label=below:{\schriftgr{$x^{1}$}}]{};
	\node[knoten] (x2) at (1.4,-4){};
	\node[knoten] (x3) at (1.4+0.5,-4)[label=below:{\schriftgr{$x^{k}$}}]{};

	\node[knoten] (a1) at (0,0)[label=below:{\schriftgr{$a^{1}$}}]{};
	\node[knoten] (a2) at (0.5,0){};
	\node[knoten] (a3) at (1,0)[label=below:{\schriftgr{$a^{k}$}}]{};

	\draw (-0.2,-0.1) rectangle (1.2,-1.9);
	\node at (0.5,-1) {$M_S^{1}$};

	\node[knoten] (c1) at (0,-2)[label=above:{\schriftgr{$c^{1}$}}]{};
	\node[knoten] (c2) at (0.5,-2){};
	\node[knoten] (c3) at (1,-2)[label=above:{\schriftgr{$c^{k}$}}]{};

	\draw[-] (y) -- (a1);
	\draw[-] (y) -- (a2);
	\draw[-] (y) -- (a3);

	\draw[-] (c1) -- (x1);
	\draw[-] (c2) -- (x2);
	\draw[-] (c3) -- (x3);

	\begin{scope}[xshift=1.8cm]
		
	\node[knoten] (b1) at (0,0)[label=below:{\schriftgr{$b^{1}$}}]{};
	\node[knoten] (b2) at (0.5,0){};
	\node[knoten] (b3) at (1,0)[label=below:{\schriftgr{$b^{k}$}}]{};

	\draw (-0.2,-0.1) rectangle (1.2,-1.9);
	\node at (0.5,-1) {$M_S^{2}$};

	\node[knoten] (d1) at (0,-2)[label=above:{\schriftgr{$d^{1}$}}]{};
	\node[knoten] (d2) at (0.5,-2){};
	\node[knoten] (d3) at (1,-2)[label=above:{\schriftgr{$d^{k}$}}]{};

	\draw[-] (y) -- (b1);
	\draw[-] (y) -- (b2);
	\draw[-] (y) -- (b3);

	\draw[-] (d1) -- (x1);
	\draw[-] (d2) -- (x2);
	\draw[-] (d3) -- (x3);
		
	\end{scope}
\end{scope}

\end{tikzpicture}

%% file: Tikz/Tikz-Choice-LICS.tex

\begin{tikzpicture}
	[scale=1, transform shape, knoten/.style={circle,draw=black,fill=black,
	inner  sep=0pt,minimum  size=1mm},kknoten/.style={circle,draw=black,fill=black,
	inner  sep=0pt,minimum  size=0.8mm}, leerknoten/.style={circle,draw=black,fill=white,
	inner  sep=0pt,minimum  size=1mm},leerkknoten/.style={circle,draw=black,fill=white,
	inner  sep=0pt,minimum  size=0.8mm}]
	
	\node at (3.1+2,-4) {Duplicator's side};
	\node at (0.75-2,-4) {Spoiler's side};
    
	\node[knoten,xshift=-2cm] (x) at (0.75,0) [label=above:{{\tiny$x^i$}}] {};

	\node[leerknoten] (x0) at (2.5+2,0) [label=above:{{\tiny$x^i_0$}}] {};
	\node[knoten] (x1) at (2.9+2,0) [label=above:{{\tiny$x^i_1$}}] {};
	\node[knoten] (x2) at (3.3+2,0)  {};
	\node[knoten] (x5) at (3.7+2,0) [label=above:{{\tiny$x^i_n$}}] {};

\begin{scope}[yshift=-3cm]

	\node[knoten,xshift=-2cm] (y1) at (0,0) [label=below:{{\tiny$(y_1)^i$}}] {};
	\node[knoten,xshift=-2cm] (y2) at (0.75,0) {};
	\node[knoten,xshift=-2cm] (y3) at (1.5,0) [label=below:{{\tiny$(y_m)^i$}}] {};

	\node[leerknoten] (y10) at (2.5,0) [label=below:{{\tiny$ (y_1)_0^i$}}] {};
	\node[knoten] (y11) at (2.9,0)  {};
	\node[knoten] (y12) at (3.3,0)  {};
	\node[knoten] (y15) at (3.7,0) [label=below:{{\tiny$(y_1)^i_n$}}] {};

	\node[leerknoten] (y20) at (2.5+2,0) {};
	\node[knoten] (y21) at (2.9+2,0){};
	\node[knoten] (y22) at (3.3+2,0)  {};
	\node[knoten] (y25) at (3.7+2,0)  {};

	\node[leerknoten] (y30) at (2.5+4,0) [label=below:{{\tiny$ (y_m)^i_0$}}] {};
	\node[knoten] (y31) at (2.9+4,0){};
	\node[knoten] (y32) at (3.3+4,0)  {};
	\node[knoten] (y35) at (3.7+4,0) [label=below:{{\tiny$(y_m)^i_n$}}] {};

\end{scope}

	\node[knoten,xshift=-2cm] (m) at (0.75,-1.5) [label=left:{{\tiny${a}^i$}}] {};
	
	\node[leerknoten] (m0) at (2.5,-1.5) [label=above:{{\tiny$a^i_0$}}] {};
	
	\node[knoten] (m11) at (2.9+0.4,-1.5) [label=above:{{\tiny${a}^i_{1,1}$}}] {};
	\node[knoten] (m12) at (3.3+0.4,-1.5){};
	\node[knoten] (m13) at (3.7+0.4,-1.5)  [label={[label distance = 1mm]35:{\tiny${a}^i_{1,m}$}}]{};
	
	\node[knoten] (m21) at (4.1+0.8,-1.5){};
	\node[knoten] (m22) at (4.5+0.8,-1.5){};
	\node[knoten] (m23) at (4.9+0.8,-1.5){};
	
	\node[knoten] (m51) at (5.3+1.2,-1.5) {};
	\node[knoten] (m52) at (5.7+1.2,-1.5){};
	\node[knoten] (m53) at (6.1+1.2,-1.5)[label=above:{{\tiny${a}^i_{n,m}$}}]{};
	
	\draw[-] (x) -- (m);
	\draw[-] (m) -- (y1);
	\draw[-] (m) -- (y2);
	\draw[-] (m) -- (y3);
	
	\foreach \x in {1,2,3} \foreach \y in {1,2,5} \draw[-] (x\y) -- (m\y\x);

	\draw[-] (x0) -- (m0);
	\draw[-] (m0) -- (y10);
	\draw[-] (m0) -- (y20);
	\draw[-] (m0) -- (y30);
	
	\draw[-] (m11) -- (y11); 
	\draw[-] (m11) -- (y20);
	\draw[-] (m11) -- (y30);
	
	\draw[-] (m12) -- (y10); 
	\draw[-] (m12) -- (y21);
	\draw[-] (m12) -- (y30);
	
	\draw[-] (m13) -- (y10); 
	\draw[-] (m13) -- (y20);
	\draw[-] (m13) -- (y31);

	\draw[-] (m21) -- (y12); 
	\draw[-] (m21) -- (y20);
	\draw[-] (m21) -- (y30);
	
	\draw[-] (m22) -- (y10); 
	\draw[-] (m22) -- (y22);
	\draw[-] (m22) -- (y30);
	
	\draw[-] (m23) -- (y10); 
	\draw[-] (m23) -- (y20);
	\draw[-] (m23) -- (y32);

	\draw[-] (m51) -- (y15); 
	\draw[-] (m51) -- (y20);
	\draw[-] (m51) -- (y30);
	
	\draw[-] (m52) -- (y10); 
	\draw[-] (m52) -- (y25);
	\draw[-] (m52) -- (y30);
	
	\draw[-] (m53) -- (y10); 
	\draw[-] (m53) -- (y20);
	\draw[-] (m53) -- (y35);
\end{tikzpicture}

%% file: LMCS-final-arxiv-code.bbl
\begin{thebibliography}{10}

\bibitem{Adachi.1984}
Akeo Adachi, Shigeki Iwata, and Takumi Kasai.
\newblock {Some combinatorial game problems require $\Omega(n^k)$ time}.
\newblock {\em J. ACM}, 31, March 1984.

\bibitem{Atserias.2007}
Albert Atserias, Andrei~A. Bulatov, and V\'{\i}ctor Dalmau.
\newblock On the power of {\it k} -consistency.
\newblock In {\em Proc. {ICALP}'07}, pages 279--290, 2007.

\bibitem{Berkholz.2012}
Christoph Berkholz.
\newblock Lower bounds for existential pebble games and k-consistency tests.
\newblock In {\em Proc. {LICS}'12}, pages 25--34, 2012.

\bibitem{Cooper.1989}
Martin C. and Cooper.
\newblock An optimal k-consistency algorithm.
\newblock {\em Artificial Intelligence}, 41(1):89 -- 95, 1989.

\bibitem{Dalmau.2002}
V\'{\i}ctor Dalmau, Phokion~G. Kolaitis, and Moshe~Y. Vardi.
\newblock Constraint satisfaction, bounded treewidth, and finite-variable
  logics.
\newblock In {\em Proc. {CP}'02}, pages 310--326, 2002.

\bibitem{Downey.1999}
Rodney~G. Downey and Michael~R. Fellows.
\newblock {\em Parameterized Complexity}.
\newblock Springer-Verlag, 1999.

\bibitem{Feder.1998}
Tom\'{a}s Feder and Moshe~Y. Vardi.
\newblock The computational structure of monotone monadic snp and constraint
  satisfaction: A study through datalog and group theory.
\newblock {\em SIAM Journal on Computing}, 28(1):57--104, 1998.

\bibitem{Gaspers.2011}
Serge Gaspers and Stefan Szeider.
\newblock The parameterized complexity of local consistency.
\newblock In {\em Proc. {CP}'11}, pages 302--316, 2011.

\bibitem{Grohe.1996}
Martin Grohe.
\newblock Equivalence in finite-variable logics is complete for polynomial
  time.
\newblock In {\em In Proceedings of the 37th Annual IEEE Symposium on
  Foundations of Computer Science}, pages 264--273, 1996.

\bibitem{Kasai.1978}
Takumi Kasai, Akeo Adachi, and Shigeki Iwata.
\newblock Classes of pebble games and complete problems.
\newblock {\em SIAM J. Comput.}, 8(4):574--586, 1979.

\bibitem{Kolaitis.2003}
Phokion~G. Kolaitis and Jonathan Panttaja.
\newblock On the complexity of existential pebble games.
\newblock In {\em Proc. {CSL}'03}, pages 314--329, 2003.

\bibitem{Kolaitis.1995}
Phokion~G. Kolaitis and Moshe~Y. Vardi.
\newblock On the expressive power of datalog: Tools and a case study.
\newblock {\em J. Comput. Syst. Sci.}, 51(1):110--134, 1995.

\bibitem{Kolaitis.2000}
Phokion~G. Kolaitis and Moshe~Y. Vardi.
\newblock Conjunctive-query containment and constraint satisfaction.
\newblock {\em J. Comput. Syst. Sci.}, 61:302--332, October 2000.

\bibitem{Kolaitis.2000a}
Phokion~G. Kolaitis and Moshe~Y. Vardi.
\newblock A game-theoretic approach to constraint satisfaction.
\newblock In {\em Proc {AAAI/IAAI}'00}, pages 175--181, 2000.

\end{thebibliography}
